\documentclass[10pt]{article}

\usepackage[english]{babel}
\usepackage[utf8]{inputenc}
\usepackage[T1]{fontenc}
\usepackage[letterpaper,top=1in,bottom=1in,left=1in,right=1in]{geometry}
\usepackage{booktabs}
\usepackage{longtable}
\usepackage{float}
\usepackage{amsmath,amssymb,amsthm,mathrsfs}
\usepackage{mathtools}
\usepackage{graphicx}
\usepackage[dvipsnames]{xcolor}
\usepackage{enumitem}
\usepackage{tikz}
\usetikzlibrary{positioning, shapes.geometric, calc}
\usepackage{algorithm}
\usepackage[noend]{algpseudocode}
\usepackage[authoryear,round]{natbib}
\usepackage[
  colorlinks=true,
  linkcolor=Violet,
  citecolor=Violet,
  urlcolor=Violet
]{hyperref}
\hypersetup{
  pdftitle={
    Towards Scaling Reinforcement Learning to Massive Populations:
    Learning Mean-Field Representations
  },
  pdfauthor={
    Aditya Makkar, Benjamin Unger, Jeongyeol Kwon,
    Mathieu Lauri\`ere, Eugene Vinitsky, Yonathan Efroni
  },
  pdfsubject={Offline mean-field reinforcement learning},
  pdfkeywords={
    reinforcement learning, mean-field games,
    offline reinforcement learning, representation learning
  },
  colorlinks=true,
  linkcolor=Violet,
  citecolor=Violet,
  urlcolor=Violet
}
\usepackage{mlmodern}

\newtheorem{theorem}{Theorem}
\newtheorem{corollary}{Corollary}
\newtheorem{lemma}{Lemma}
\newtheorem{proposition}{Proposition}
\theoremstyle{definition}

\newtheorem{remark}{Remark}
\theoremstyle{plain}
\newtheorem{assumption}{Assumption}

\newcommand{\EE}{\mathbb{E}}
\newcommand{\NN}{\mathbb{N}}
\newcommand{\PP}{\mathbb{P}}
\newcommand{\RR}{\mathbb{R}}
\newcommand{\cA}{\mathcal{A}}
\newcommand{\cB}{\mathcal{B}}
\newcommand{\cC}{\mathcal{C}}
\newcommand{\cL}{\mathcal{L}}
\newcommand{\cR}{\mathcal{R}}

\newcommand{\cX}{\mathcal{X}}

\newcommand{\eopt}{\varepsilon_{\mathrm{opt}}}
\newcommand{\fp}{\mathfrak{p}}

\newcommand{\eps}{\varepsilon}

\newcommand{\lr}[1]{\left( #1 \right)}

\providecommand{\coloneq}{\coloneqq}
\DeclareMathOperator{\br}{BR}
\DeclareMathOperator{\gap}{Gap}
\newcommand{\tv}{D_{\mathrm{TV}}}
\newcommand{\kl}{D_{\mathrm{KL}}}
\newcommand{\hel}{D_{\mathrm{H}}}
\DeclareMathOperator*{\argmax}{arg\,max}

\title{Towards Scaling Reinforcement Learning to Massive Populations:\\Learning Mean-Field Representations}
\author{
Aditya Makkar\textsuperscript{1} \quad
Benjamin Unger\textsuperscript{2} \quad
Jeongyeol Kwon\textsuperscript{3}\\
Mathieu Lauri\`ere\textsuperscript{1} \quad
Eugene Vinitsky\textsuperscript{1} \quad
Yonathan Efroni\textsuperscript{4}\\[0.5em]
\small \textsuperscript{1}New York University \quad
\small \textsuperscript{2}ETH Zurich \quad
\small \textsuperscript{3}Meta Platforms Inc. \quad
\small \textsuperscript{4}Tel Aviv University
}
\date{}

\begin{document}
\maketitle

\begin{abstract}
Modern multi-agent systems are increasingly deployed at scale over large populations of agents in settings such as ad-auctions, traffic routing, and recommendation systems.
The dominant approach in such settings is to optimize each agent's policy independently, treating the other agents as part of a fixed single-agent environment rather than modeling the population dynamics.
In many large-population systems, the dynamics depend on an aggregate summary of the population rather than the identity of any individual. Mean-field RL exploits such structure, providing a principled framework that models each agent's environment as an explicit function of the population distribution. However, in large state-action spaces or high-dimensional control problems, modeling the population distribution is itself intractable. \emph{How can we design a scalable framework for high-dimensional control problems with large populations?}
This work explores this question from the perspective of representation learning. We introduce a mean-field RL framework in which the rewards and transition dynamics depend on the population only through an unknown low-dimensional aggregate statistic. We then study this framework in the offline setting and design a provable approach that learns a near-optimal policy by learning a low-dimensional representation.
Motivated by real-life supply-chain optimization problems, we design a one-step routing game to test the hypothesis that learning a low-dimensional population representation improves reward prediction and Nash gap estimation relative to baselines that don't exploit this structure.
We show that under a fixed neural-network parameter count and optimization budget, learning a low-dimensional population representation improves reward prediction and the equilibrium quality of the resulting policies.
\end{abstract}

\section{Introduction}
\label{sec:intro}

Many modern decision systems involve large populations of strategic agents, from traffic routing and power-grid markets to recommendation and resource-allocation platforms.
Directly modeling all agents jointly is statistically and computationally intractable, since the joint state-action space grows exponentially with the population size.
Mean-field games (MFGs; \citealp{huang2006large,lasry2007mean,carmona2018probabilistic}) provide a tractable approximation for such large-population multi-agent systems by replacing the interactions among finitely many agents with the interaction of a representative agent with a limiting population law, called the mean-field distribution.
This paradigm has also led to learning-based approaches for large population games, ranging from adaptive control and approximate dynamic programming to reinforcement learning (RL) methods \citep{kizilkale2013meanfield,yin2014learningMFG,yang2018meanfieldmarl,subramanian2019reinforcement}; see \citet{lauriere2022learning} for a survey, as well as the references therein.

Mean-field models make large-population games tractable by replacing the joint configuration of all agents with the mean-field distribution.
However, this distribution lives in a \((|\mathcal X||\mathcal A|-1)\)-dimensional simplex, where $|\mathcal X|$ and $|\mathcal A|$ are the cardinalities of the state and action spaces, respectively.
Consequently, representing this distribution and learning how rewards and transitions vary over it becomes statistically and computationally prohibitive in large state-action spaces.
The key observation motivating this work is that the rewards and transitions often depend on the mean-field distribution only through a much lower-dimensional statistic (see Figure~\ref{fig:representation_illustration}).
This raises the central question: \textit{How can we exploit such structure to design efficient learning algorithms that scale to high-dimensional decision problems with large numbers of agents?}

In this work, we study this question through the lens of representation learning in MFGs. We formalize a framework for learning low-dimensional representations and focus on the offline setting, where the representation must be learned from logged representative-agent trajectories together with finite-population snapshots. The offline equilibrium problem presents a coupled statistical challenge. Logged data identify rewards and transition dynamics only along behavior-induced population flows. However, a learned policy induces a different population flow, while unilateral deviations generate different representative-agent occupancy measures. Consequently, the learner must construct representations and models that generalize across both population flows and deviations, and then plan in a way that guarantees low exploitability under partial coverage.

Our analysis is closest to \citet{zhang2023offline}, who study offline equilibrium learning in Markov games.
Their algorithm constructs interval estimates of policy values by evaluating possible responses to a policy optimistically and evaluating the policy itself pessimistically, then minimizing the resulting upper confidence bound on the equilibrium gap.
Their guarantee is adaptive because, instead of requiring coverage of every unilateral deviation from the policy being evaluated, it requires coverage of that policy itself and, for each player, of an approximate response whose coverage is traded off against its best-response suboptimality.
We use the same interval-certification idea for mean-field RL, but coverage must now transfer across representative-agent occupancies, population flows, and learned representation paths; because the learner sees only finite samples from each logged population rather than the exact population law, the guarantee also includes an error term for estimating the mean-field representation.

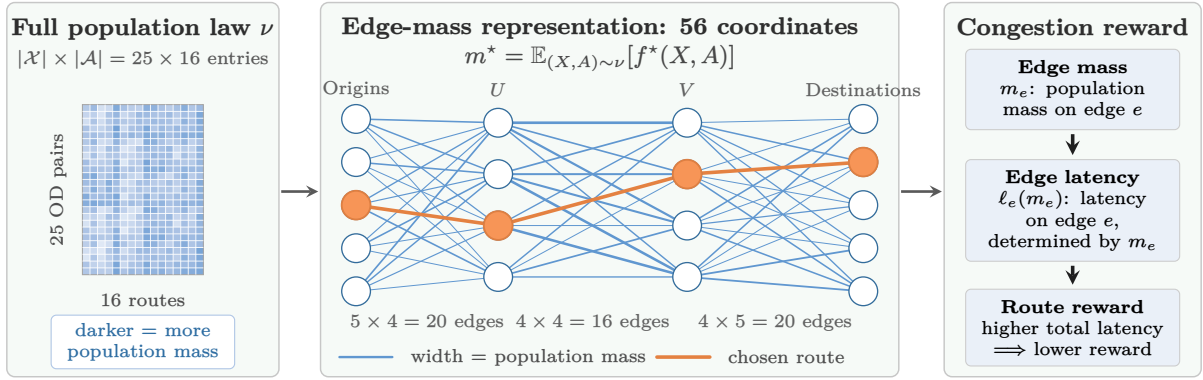
\begin{figure}[t]
\centering
\begin{tikzpicture}[
  x=1cm,
  y=1cm,
  font=\scriptsize,
  panel/.style={draw=black!22, fill=black!1, rounded corners=3pt,
    line width=0.55pt},
  netnode/.style={circle, draw=RoyalBlue!70!black, fill=white,
    minimum size=3.8mm, inner sep=0pt, line width=0.5pt},
  route node/.style={netnode, draw=Orange!90!black, fill=Orange!82,
    line width=0.7pt},
  mass edge/.style={draw=RoyalBlue!62, line cap=round},
  focal edge/.style={->, >=stealth, draw=Orange!92!black,
    line width=1.35pt, line cap=round},
  flow arrow/.style={->, >=stealth, draw=black!65, line width=1.0pt},
  down arrow/.style={->, >=stealth, draw=black!88, line width=1.35pt},
  concept box/.style={draw=black!20, fill=white, rounded corners=2pt,
    text width=2.55cm, inner sep=4pt, align=center}
]
\path[use as bounding box] (0,0) rectangle (15.9,5.15);
\path[panel, fill=ForestGreen!3] (0.05,0.10) rectangle (3.65,5.05);
\path[panel, fill=ForestGreen!3] (4.20,0.10) rectangle (11.85,5.05);
\path[panel, fill=ForestGreen!3] (12.45,0.10) rectangle (15.85,5.05);

\node[font=\bfseries\small, text=black!82] at (1.85,4.70)
  {Full population law \(\nu\)};
\node[text=black!65] at (1.85,4.30)
  {\(|\cX| \times |\cA| = 25\times16\) entries};
\foreach \col in {1,...,16}
  \foreach \row in {1,...,25} {
    \pgfmathtruncatemacro{\originidx}{floor((\row-1)/5)+1}
    \pgfmathtruncatemacro{\destidx}{mod(\row-1,5)+1}
    \pgfmathtruncatemacro{\uidx}{floor((\col-1)/4)+1}
    \pgfmathtruncatemacro{\vidx}{mod(\col-1,4)+1}
    \pgfmathtruncatemacro{\oubias}{mod(
      17*\originidx+31*\uidx+11*\originidx*\uidx
      +7*\originidx*\originidx+5*\uidx*\uidx,23)}
    \pgfmathtruncatemacro{\uvbias}{mod(
      13*\uidx+29*\vidx+17*\uidx*\vidx
      +3*\uidx*\uidx+11*\vidx*\vidx,23)}
    \pgfmathtruncatemacro{\vdbias}{mod(
      19*\vidx+23*\destidx+7*\vidx*\destidx
      +5*\vidx*\vidx+13*\destidx*\destidx,23)}
    \pgfmathtruncatemacro{\cellnoise}{mod(
      31*\row+17*\col+7*\row*\col
      +11*\row*\row+13*\col*\col,17)}
    \pgfmathtruncatemacro{\tone}{6+floor(
      (4*\oubias+4*\uvbias+4*\vdbias+\cellnoise)/5)}
    \fill[RoyalBlue!\tone]
      ({1.05+0.10*(\col-1)},{1.45+0.09*(\row-1)})
      rectangle ++(0.10,0.09);
  }
\draw[black!35, line width=0.4pt] (1.05,1.45) rectangle (2.65,3.70);
\foreach \k in {1,...,15}
  \draw[white, opacity=0.88, line width=0.12pt]
    ({1.05+0.10*\k},1.45) -- ({1.05+0.10*\k},3.70);
\foreach \k in {1,...,24}
  \draw[white, opacity=0.88, line width=0.12pt]
    (1.05,{1.45+0.09*\k}) -- (2.65,{1.45+0.09*\k});
\node[text=black!70] at (1.85,1.12) {\(16\) routes};
\node[text=black!70, rotate=90] at (0.73,2.575) {\(25\) OD pairs};
\node[draw=RoyalBlue!30, fill=white, rounded corners=2pt,
  inner xsep=6pt, inner ysep=3pt, text=RoyalBlue!75!black, align=center]
  at (1.85,0.56) {darker \(=\) more\\population mass};

\draw[flow arrow] (3.68,2.55) -- (4.17,2.55);

\node[font=\bfseries\small, text=black!82] at (7.90,4.72)
  {Edge-mass representation: 56 coordinates};
\node[font=\small, text=black!72] at (7.90,4.32)
  {\(m^\star=\EE_{(X,A)\sim\nu}[f^\star(X,A)]\)};
\node[text=black!62] at (4.67,3.90) {Origins};
\node[text=black!62] at (6.55,3.90) {\(U\)};
\node[text=black!62] at (9.05,3.90) {\(V\)};
\node[text=black!62] at (11.38,3.90) {Destinations};

\foreach \i/\yy in {1/1.23,2/1.80,3/2.37,4/2.94,5/3.51}
  \coordinate (Oc\i) at (4.67,\yy);
\foreach \u/\yy in {1/1.42,2/2.10,3/2.78,4/3.46}
  \coordinate (Uc\u) at (6.55,\yy);
\foreach \v/\yy in {1/1.42,2/2.10,3/2.78,4/3.46}
  \coordinate (Vc\v) at (9.05,\yy);
\foreach \j/\yy in {1/1.23,2/1.80,3/2.37,4/2.94,5/3.51}
  \coordinate (Dc\j) at (11.38,\yy);

\foreach \i/\u/\w in {
    1/1/0.18, 1/2/0.44, 1/3/0.55, 1/4/0.64,
    2/1/0.24, 2/2/0.24, 2/3/0.60, 2/4/0.43,
    3/1/0.62, 3/2/0.36, 3/3/0.44, 3/4/0.52,
    4/1/0.29, 4/2/0.28, 4/3/0.59, 4/4/0.41,
    5/1/0.25, 5/2/0.48, 5/3/0.56, 5/4/0.59}
  \draw[mass edge, line width=\w pt] (Oc\i) -- (Uc\u);
\foreach \u/\v/\w in {
    1/1/0.51, 1/2/0.65, 1/3/0.30, 1/4/0.55,
    2/1/0.97, 2/2/0.31, 2/3/0.41, 2/4/0.52,
    3/1/0.97, 3/2/0.79, 3/3/0.75, 3/4/0.65,
    4/1/0.55, 4/2/0.84, 4/3/0.62, 4/4/1.00}
  \draw[mass edge, line width=\w pt] (Uc\u) -- (Vc\v);
\foreach \v/\j/\w in {
    1/1/0.68, 1/2/0.69, 1/3/0.26, 1/4/0.37, 1/5/0.58,
    2/1/0.59, 2/2/0.26, 2/3/0.47, 2/4/0.27, 2/5/0.59,
    3/1/0.59, 3/2/0.39, 3/3/0.27, 3/4/0.20, 3/5/0.20,
    4/1/0.52, 4/2/0.47, 4/3/0.50, 4/4/0.58, 4/5/0.24}
  \draw[mass edge, line width=\w pt] (Vc\v) -- (Dc\j);

\draw[focal edge] (Oc3) -- (Uc2);
\draw[focal edge] (Uc2) -- (Vc3);
\draw[focal edge] (Vc3) -- (Dc4);
\foreach \i in {1,...,5} \node[netnode] at (Oc\i) {};
\foreach \u in {1,...,4} \node[netnode] at (Uc\u) {};
\foreach \v in {1,...,4} \node[netnode] at (Vc\v) {};
\foreach \j in {1,...,5} \node[netnode] at (Dc\j) {};
\node[route node] at (Oc3) {};
\node[route node] at (Uc2) {};
\node[route node] at (Vc3) {};
\node[route node] at (Dc4) {};
\node[fill=ForestGreen!3, inner sep=1pt, text=black!62] at (5.61,0.82)
  {\(5\times4=20\) edges};
\node[fill=ForestGreen!3, inner sep=1pt, text=black!62] at (7.80,0.82)
  {\(4\times4=16\) edges};
\node[fill=ForestGreen!3, inner sep=1pt, text=black!62] at (10.215,0.82)
  {\(4\times5=20\) edges};
\draw[mass edge, line width=0.90pt] (4.45,0.36) -- (5.15,0.36);
\node[anchor=west, text=black!70] at (5.27,0.36)
  {width \(=\) population mass};
\draw[Orange!92!black, line width=1.35pt, line cap=round]
  (8.65,0.36) -- (9.35,0.36);
\node[anchor=west, text=black!70] at (9.47,0.36)
  {chosen route};

\draw[flow arrow] (11.88,2.55) -- (12.42,2.55);
\node[font=\bfseries\small, text=black!82] at (14.15,4.70)
  {Congestion reward};
\node[concept box, fill=RoyalBlue!7] (massbox) at (14.15,3.90)
  {\textbf{Edge mass}\\\(m_e\): population mass on edge \(e\)};
\node[concept box, fill=RoyalBlue!7] (latencybox) at (14.15,2.30)
  {\textbf{Edge latency}\\\(\ell_e(m_e)\): latency on edge \(e\),\\determined by \(m_e\)};
\node[concept box, fill=RoyalBlue!7] (rewardbox) at (14.15,0.75)
  {\textbf{Route reward}\\higher total latency\\\(\Longrightarrow\) lower reward};
\draw[down arrow, shorten <=1.5pt, shorten >=1.5pt]
  (massbox.south) -- (latencybox.north);
\draw[down arrow, shorten <=1.5pt, shorten >=1.5pt]
  (latencybox.south) -- (rewardbox.north);
\end{tikzpicture}
\caption{
Illustration of the one-step routing game and its mean-field representation.
There are 5 origins, 5 destinations and two intermediate layers containing 4 nodes each.
The state fixes an origin--destination pair, while the action selects one node in each intermediate layer, producing a three-edge route (orange).
Therefore, \(|\cX|\times|\cA| = 25 \times 16 = 400.\)
But since there are only 20+16+20 = 56 directed edges, the mean-field representation dimension can be 56.
Here, \(f^\star(X,A) \in \{0,1\}^{56}\) records which network edges the corresponding route uses, and is unknown to the learner.
Blue edge widths are proportional to masses aggregated from the illustrated population law.
The population mass \(m_e\) on edge \(e\) determines its latency \(\ell_e(m_e)\), and higher total route latency gives lower reward.}
\label{fig:representation_illustration}
\end{figure}

\subsection{Contributions}
Our approach to the problem discussed above is model-based: it learns a representation-dependent reward and transition model and uses an optimistic--pessimistic interval objective to evaluate candidate equilibria and plausible deviations. Our main contributions are three-fold:
\begin{enumerate}
  \item We introduce a \textbf{representation-learning framework for MFGs} in which rewards and transitions depend on the population law only through an unknown low-dimensional statistic. Unlike standard single-agent representation learning, where the representation usually encodes the agent's own state or state-action pair, the representation here summarizes the surrounding population and changes with the population policy. This statistic must be learned from representative-agent trajectories and finite population samples.
  \item We develop a \textbf{model-based algorithm} for the \textbf{offline setting} that learns the reward model, the transition model, and importantly the mean-field representations. We prove a finite-sample Nash-gap guarantee under an adaptive coverage condition.
  \item We evaluate the representation-learning mechanism in a \textbf{one-step congestion-routing game} motivated by fulfillment-routing problems arising in supply-chain optimization (see Figure~\ref{fig:representation_illustration}). Under a fixed neural-network parameter count and optimization budget, we compare our model with baselines that either do not fully exploit the population structure or have privileged access to the true mean field. Our model outperforms the former and approaches the latter in the large-data regime.
\end{enumerate}

\subsection{Notation}
For a finite set \(\cX\), write \(\Delta(\cX)\) for the set of probability distributions on $\cX$.
For \(p,q\in\Delta(\cX)\), write
\[
\tv(p,q) \coloneq \frac{1}{2}\sum_{x\in\cX}|p(x)-q(x)|,
\quad
\hel^2(p,q) \coloneq \frac12\sum_{x\in\cX}\bigl(\sqrt{p(x)}-\sqrt{q(x)}\bigr)^2,
\quad
\kl(p\|q) \coloneq \sum_{x\in\cX}p(x)\log\frac{p(x)}{q(x)},
\]
with the conventions \(0\log(0/b)=0\) for \(b \ge 0\), and \(a\log(a/0)=+\infty\) for \(a>0\).
We consider finite state and action spaces \(\cX\) and \(\cA\), a finite horizon \(H\), and an initial state distribution \(\mu_0\in\Delta(\mathcal X)\).
A (non-stationary Markov) policy \(\pi\) is a sequence \(\pi=(\pi_h)_{h=0}^{H-1}\) with \(\pi_h \colon \cX \to \Delta(\cA)\).
The set of all such policies is denoted by \(\Pi\).
For a state law \(\mu \in \Delta(\cX)\) and a decision rule \(\alpha \colon \cX \to \Delta(\cA)\), define the induced state-action law \(\mu \otimes \alpha \in \Delta(\cX \times \cA)\) as \((\mu\otimes\alpha)(x,a) \coloneq \mu(x) \alpha(a\mid x)\).

\section{Mean-Field Games with Low-Dimensional Representations}
\label{sec:mfg-repr-setting}

In finite-population Markov games, the state of the environment includes the joint configuration of all agents, which is exponentially large in the population size.
Mean-field games replace this joint configuration by the limiting population law.
In other words, a representative agent interacts with the population only through a distribution on \(\cX\times\cA\), rather than through the identities and states of all other agents.
However, as discussed in the introduction, the full mean-field distribution can still be too large to model or learn when \(|\cX||\cA|\) is large.
Fortunately, many large-population systems have an additional structure where the rewards and transitions depend on the mean-field distribution only through a low-dimensional aggregate statistic, such as the density of agents in a small number of regions.

In this section, we first rewrite the standard representative-agent control problem using this low-dimensional representation.
Second, we define the corresponding value functions, best-response, and Nash-gap objects and show that the induced game admits a Nash equilibrium under mild continuity assumptions.
Third, we note that while the one-step reward and transition at stage \(h\) depend on the current representation \(m_h\), the value functions are indexed by the representation suffix \(m_{h:H-1}\).

\subsection{Mean-Field Game as a Representative-Agent Control Problem}
Mean-field games describe the limiting strategic interaction faced by a representative agent in a large symmetric population.
In this limit, an individual agent is negligible: changing one agent's policy does not change the population flow.
The population enters a representative agent's reward and transition only through a population law.

More concretely, there is a stage-wise reward function $\bar r^\star_h$ and a transition kernel $\bar T_h^\star$, dictating the rewards and transitions for each representative agent, with the following form: for all \(h \in \{0, \ldots, H-1\}\),
\[
\bar r_h^\star \colon \cX \times \cA \times \Delta(\cX \times \cA) \to [0, 1], \qquad
\bar T_h^\star(\cdot \mid x_h, a_h, \nu_h) \in \Delta(\cX) \; \forall x_h \in \cX,\, a_h \in \cA,\, \nu_h \in \Delta(\cX \times \cA).
\]
They are functions of the population state-action distribution \(\nu_h\).

When the population uses policy \(\eta \in \Pi\), its state and state-action laws are defined recursively by
\[
\mu_0^{\star,\eta}=\mu_0, \qquad
\nu_h^{\star,\eta} = \mu_h^{\star,\eta}\otimes \eta_h, \qquad
\mu_{h+1}^{\star,\eta}(x') = \sum_{x\in\cX} \sum_{a\in\cA} \nu_h^{\star,\eta}(x,a) \bar T_h^\star(x'\mid x,a,\nu_h^{\star,\eta}).
\]
If a representative agent instead plays \(\pi\in\Pi\) against a population playing \(\eta\), the population flow \((\nu_h^{\star,\eta})_{h=0}^{H-1}\) remains fixed since a single agent's deviation has negligible impact on the population flow.
Thus, the state-action process \((X_0, A_0, \ldots, X_{H-1}, A_{H-1}, X_H)\) for the representative agent is
\begin{equation}
\label{eq:representative-agent-process}
X_0\sim\mu_0, \qquad
A_h\sim\pi_h(\cdot\mid X_h), \qquad
X_{h+1} \sim \bar T_h^\star(\cdot\mid X_h,A_h,\nu_h^{\star,\eta}), \qquad h \in \{0, \ldots, H-1\}.
\end{equation}
In words, against any fixed population flow, the representative agent faces a single-agent control problem.
The game-theoretic part enters through the requirement that \(\eta\) be optimal against the population flow it itself induces.

\subsection{Low-Dimensional Mean-field representations}
The full law \(\nu_h^{\star,\eta}\) lies in a simplex of dimension \(|\cX||\cA|-1\).
We assume that the game depends on this law only through a low-dimensional mean-field representation.
Specifically, there exists \(d \in \NN\) such that, for each \(h \in \{0, \ldots, H-1\}\), there are a feature map, a reward map and a transition kernel
\[
f_h^\star \colon \cX \times \cA \to \RR^d, \qquad
r_h^\star \colon \cX \times \cA \times \RR^d \to [0, 1], \qquad
T_h^\star(\cdot \mid x, a, m) \in \Delta(\cX) \; \forall x \in \cX,\, a \in \cA,\, m \in \RR^d,
\]
such that, for every \(x_h \in \cX,\, a_h \in \cA,\, \nu_h\in\Delta(\cX\times\cA)\),
\[
\bar r_h^\star(x_h,a_h,\nu_h) = r_h^\star\!\lr{x_h,a_h,m_h^\star}, \qquad
\bar T_h^\star(\cdot\mid x_h,a_h,\nu_h) = T_h^\star\!\lr{\cdot\mid x_h,a_h,m_h^\star},
\]
where \[m_h^\star \coloneq \mathbb{E}_{(x,a)\sim\nu_h}\left[f_h^\star(x,a)\right]=\sum_{x\in\cX} \sum_{a\in\cA} \nu_h(x,a)f_h^\star(x,a) \in \RR^d.\]
Thus, rewards and transitions depend on a potentially lower-dimensional representation instead of the full population distribution.
Ideally, \(d \ll |\cX||\cA|-1\).

The induced mean-field representation under population policy \(\eta\) is
\[m_h^{\star,\eta} \coloneq \sum_{x\in\cX} \sum_{a\in\cA} \nu_h^{\star,\eta}(x,a) f_h^\star(x,a).\]
Under this structure, the population dynamics induced by \(\eta\) can be written as follows
\[
    \mu_0^{\star,\eta} = \mu_0, \qquad
    \nu_h^{\star,\eta} = \mu_h^{\star,\eta}\otimes\eta_h, \qquad
    \mu_{h+1}^{\star,\eta}(x') = \sum_{x,a} \nu_h^{\star,\eta}(x,a) T_h^\star(x'\mid x,a,m_h^{\star,\eta}).
\]
For a representative agent playing \(\pi\) against \(\eta\), let \(\PP_\star^{\pi;\eta}\) denote the path law defined as in \eqref{eq:representative-agent-process}, with the full-law transition \(\bar T_h^\star(\cdot\mid X_h,A_h,\nu_h^{\star,\eta})\) replaced by its representation form \(T_h^\star(\cdot\mid X_h,A_h,m_h^{\star,\eta})\).
Let \(\EE_\star^{\pi;\eta}\) denote the expectation operator under this probability measure.

The value of \(\pi\) against \(\eta\) is defined as
\[
    J(\pi;\eta) \coloneq \mathbb E_\star^{\pi;\eta} \left[\sum_{h=0}^{H-1} r_h^\star(X_h,A_h,m_h^{\star,\eta})\right].
\]
We write \(\mu_h^{\star,\pi;\eta}\) for the representative agent's state law under \(\mathbb P_\star^{\pi;\eta}\), and \(\nu_h^{\star,\pi;\eta} \coloneq \mu_h^{\star,\pi;\eta}\otimes\pi_h\) for the representative agent's state-action law.

The classical full-law mean-field game is recovered by taking \(d=|\mathcal X||\mathcal A|\) and letting \(f_h^\star(x,a)\) be the standard basis vector indexed by \((x,a)\).
Then \(m_h^{\star,\eta} = \nu_h^{\star,\eta}\).

\subsection{Value Functions Depend on Future Representations}

Fix a representative-agent policy \(\pi\in\Pi\).
A population policy \(\eta\in\Pi\) induces the representation path \(m_{0:H-1}^{\star,\eta} = \bigl(m_0^{\star,\eta},\ldots,m_{H-1}^{\star,\eta}\bigr)\).
Because a unilateral deviation does not affect the population flow, this path remains fixed from the representative agent's perspective.
Once the path is fixed, the representative agent therefore faces a finite-horizon MDP whose reward and transition kernel at stage \(h\) are \(r_h^\star(\cdot,\cdot,m_h^{\star,\eta})\) and \(T_h^\star(\cdot\mid\cdot,\cdot,m_h^{\star,\eta})\).

To make the dependence on the representation path explicit, fix an arbitrary suffix \(z_{h:H-1}=(z_h,\ldots,z_{H-1})\in(\RR^d)^{H-h}\).
For a fixed policy \(\pi\), define the value functions by backward recursion.
Set \(v_H^{\star,\pi}(x,\varnothing)\coloneq 0\), and, for \(h=H-1,\ldots,0\), define
\begin{equation}
\label{eq:suffix-value-recursion}
\begin{aligned}
q_h^{\star,\pi}(x,a,z_{h:H-1}) &\coloneq r_h^\star(x,a,z_h) + \sum_{x'\in\cX} T_h^\star(x'\mid x,a,z_h) v_{h+1}^{\star,\pi}(x',z_{h+1:H-1}),\\
v_h^{\star,\pi}(x,z_{h:H-1}) &\coloneq \sum_{a\in\cA} \pi_h(a\mid x) q_h^{\star,\pi}(x,a,z_{h:H-1}).
\end{aligned}
\end{equation}
Here \(z_{H:H-1}\) denotes the empty suffix.

Thus, when the population uses \(\eta\), the representative agent's action and state values are obtained by evaluating these functions along the induced representation suffix to get \(q_h^{\star,\pi}\left(x,a,m_{h:H-1}^{\star,\eta}\right)\) and \(v_h^{\star,\pi}\left(x,m_{h:H-1}^{\star,\eta}\right)\).
In particular, \(J(\pi;\eta) = \sum_{x\in\cX} \mu_0(x) v_0^{\star,\pi} \left(x,m_{0:H-1}^{\star,\eta}\right)\).

The form of the value functions makes clear that even though the current representation \(z_h\) determines the immediate reward and next-state distribution, the value functions depend on the future suffix \(z_{h+1:H-1}\).
Consequently, changing future representations can change \(q_h^{\star,\pi}\) while leaving \((x,a,z_h)\) unchanged.

\begin{proposition}[The current representation need not determine the action value]
\label{prop:one-step-representation-not-sufficient}
There exists a horizon-two mean-field model satisfying the representation structure above, together with policies \(\eta,\tilde\eta,\pi\in\Pi\), a state \(x\in\cX\), and an action \(a\in\cA\), such that \(m_0^{\star,\eta}=m_0^{\star,\tilde\eta}\), but \(q_0^{\star,\pi}\left(x,a,m_{0:1}^{\star,\eta}\right) \neq q_0^{\star,\pi}\left(x,a,m_{0:1}^{\star,\tilde\eta}\right)\).
Consequently, action value function cannot in general be represented uniformly over population policies as functions of \((x,a,m_h^{\star,\eta})\) alone.
\end{proposition}

The proof uses the fact that two population policies may agree on the current representation but induce different representations at later stages, where the agent receives different rewards or faces different transitions; see App.~\ref{app:one-step-proof}.
This is why the Bellman residuals in Section~\ref{sec:theory} evaluate value functions along complete representation suffixes.

\subsection{Equilibrium and Well-Posedness}
Finally, we study the notion of equilibrium in the setting introduced above, and show it is well posed.
For a population policy \(\eta\), define the best-response set and the exploitability, or Nash gap, by
\[
\br(\eta) \coloneq \argmax_{\pi\in\Pi} J(\pi;\eta),\quad
\gap(\eta) \coloneq \sup_{\pi\in\Pi} J(\pi;\eta)-J(\eta;\eta).
\]
A policy \(\eta^\star\in\Pi\) is a Nash equilibrium if \(\eta^\star\in\br(\eta^\star)\), or equivalently, \(\gap(\eta^\star)=0\).
Let \(\Pi_\star\coloneq\{\eta\in\Pi:\gap(\eta)=0\}\) denote the set of Nash equilibria.
This set is non-empty:

\begin{proposition}[Existence of Nash equilibria]
\label{prop:existence-nash}
Suppose \(\mathcal X\) and \(\mathcal A\) are finite, each feature map \(f_h^\star\) is bounded, and for every \(h,x,a\), the maps \(m\mapsto r_h^\star(x,a,m)\) and \(m\mapsto T_h^\star(\cdot\mid x,a,m)\) are continuous.
Then there exists a policy \(\eta^\star\in\Pi\) such that \(\gap(\eta^\star)=0\).
\end{proposition}

The proof is deferred to App.~\ref{app:existence-proof}.
It uses the standard occupancy-measure formulation and a Kakutani fixed-point argument applied to the mean-field representation path.

We note here that while our focus is on the concept of Nash equilibrium, a similar approach could be used for social optimum (also known as mean field control) where the objective is \(\max_{\eta \in \Pi}\; J(\eta;\, \eta)\).

\section{Offline Mean-Field RL Framework}
\label{sec:offline}

Section~\ref{sec:mfg-repr-setting} formalized a class of mean-field games in which rewards and transitions depend on a low-dimensional statistic rather than the full population law. In this section, we define an offline learning setting in which the mean-field representation is unknown but belongs to a function class a learner can access.

We introduce the model class, the offline observation model, and the empirical criterion used to learn a model.
The offline data contains independent trajectory data and population samples for each stage.
The main statistical and planning guarantees are stated in Section~\ref{sec:theory}.

\subsection{Candidate mean-field models}
A candidate model is a tuple \(\theta=(f^\theta,r^\theta,T^\theta)\in\Theta\), where \(f_h^\theta\colon\cX\times\cA\to\RR^d\) is a feature map, \(r_h^\theta\colon\cX\times\cA\times\RR^d\to[0,1]\) is a reward model, and \(T_h^\theta(\cdot\mid x,a,m)\in\Delta(\cX)\) is a transition model.
We assume that \(\Theta\) is finite.
For a population policy \(\eta\), the model \(\theta\) induces the population flow \((\mu_h^{\theta, \eta}, \nu_h^{\theta, \eta}, m_h^{\theta, \eta})_{h=0}^{H-1}\) just like in the true model case discussed above with \(\star\) replaced with \(\theta\).
For future reference, define
\[\bar m_h^{\theta,\eta} \coloneq \sum_{x \in \cX} \sum_{a \in \cA} \nu_h^{\star,\eta}(x,a)f_h^\theta(x,a).\]
Against population policy \(\eta\), a representative agent playing \(\pi\) under model \(\theta\) follows the path law \(\PP_\theta^{\pi;\eta}\) defined as in \eqref{eq:representative-agent-process}, with transition \(T_h^\theta(\cdot\mid X_h,A_h,m_h^{\theta,\eta})\).
The corresponding candidate value is
\[
J_\theta(\pi;\eta) \coloneq
\EE_\theta^{\pi;\eta}
\left[
\sum_{h=0}^{H-1} r_h^\theta(X_h,A_h,m_h^{\theta,\eta})
\right].
\]

We work with the following assumptions.
\begin{assumption}[Realizability]
\label{ass:realizability}
The true model belongs to the candidate class: \(\theta^\star=(f^\star,r^\star,T^\star)\in\Theta\).
\end{assumption}

\begin{assumption}[Boundedness]
\label{ass:boundedness}
There are constants \(B_f<\infty\) and \(\tau\in(0,1]\).
Let \(\cB_f\coloneq\{m\in\RR^d:\|m\|_2\le B_f\}\).
For every \(\theta\in\Theta\), \(h\in\{0,\ldots,H-1\}\), \((x,a)\in\cX\times\cA\), \(x'\in\cX\), and \(m\in\cB_f\),
\[
\|f_h^\theta(x,a)\|_2\le B_f,
\qquad
T_h^\theta(x'\mid x,a,m)\in\{0\}\cup[\tau,1].
\]
\end{assumption}

\begin{assumption}[Continuity]
\label{ass:continuity}
There are constants \(L_{r,h}<\infty\) and \(L_{T,h}<\infty\) such that, for every \(\theta\in\Theta\), \(h\in\{0,\ldots,H-1\}\), \((x,a)\in\cX\times\cA\), and \(m,m'\in\cB_f\),
\[
|r_h^\theta(x,a,m)-r_h^\theta(x,a,m')|
\le L_{r,h}\|m-m'\|_2,
\]
and
\[
\tv\!\left(
T_h^\theta(\cdot\mid x,a,m),
T_h^\theta(\cdot\mid x,a,m')
\right)
\le
L_{T,h}\|m-m'\|_2.
\]
\end{assumption}

\subsection{Offline data}
We observe \(N \in \NN\) independent logged episodes with each sample containing a pair of the representative player's trajectory and \(K \in \NN\) independent samples of population snapshots over time: \[\{(X_{n,h},A_{n,h},R_{n,h},X_{n,h+1},(X_{n,h}^{(k)},A_{n,h}^{(k)})_{k=1}^K)_{h=0}^{H-1}\}_{n=1}^N.\]
Here, for each $n$, a behavior population policy \(\rho_n\) is drawn independently from a distribution \(\fp\) over \(\Pi\), and conditional on \(\rho_n\), the representative trajectory \((X_{n,0}, A_{n,0}, \ldots, X_{n,H-1}, A_{n,H-1}, X_{n,H})\) is drawn from \(\PP_\star^{\rho_n;\rho_n}\), and the observed reward is \(R_{n,h}=r_h^\star(X_{n,h},A_{n,h},m_h^{\star,\rho_n})\).
Conditional on \(\rho_n\), the entire collection of population snapshots is independent of the representative trajectory.
At each time-step \(h\), the snapshot consists of \(K\) samples \((X_{n,h}^{(k)},A_{n,h}^{(k)})\stackrel{\mathrm{i.i.d.}}{\sim} \nu_h^{\star,\rho_n},\).

The sample can be used to compute the empirical mean-field representation for a candidate model \(\theta\): \[\widehat m_{n,h}^{\theta} \coloneq \frac{1}{K}\sum_{k=1}^K f_h^\theta(X_{n,h}^{(k)},A_{n,h}^{(k)}).\]
Conditionally on \(\rho_n\), it satisfies \(\EE[\widehat m_{n,h}^{\theta}\mid \rho_n] = \bar m_h^{\theta,\rho_n}\).
Note that \(\widehat m_{n,h}^{\theta}\), \(\bar m_h^{\theta,\rho_n}\), and \(m_h^{\theta,\rho_n}\) play different roles.
The empirical risk can use only \(\widehat m_{n,h}^{\theta}\), the population prediction loss is naturally expressed at \(\bar m_h^{\theta,\rho_n}\), and planning under model \(\theta\) uses the representation \(m_h^{\theta,\eta}\).

\subsection{Learning algorithm}
Algorithm~\ref{alg:offline-representation-learning} summarizes the offline learning procedure. For each model \(\theta\), we use the finite-population samples to construct empirical representation \(\widehat m^\theta\) and fit the reward and transition components using the empirical risk in \eqref{eq:empirical-risk}. We then retain the near-minimizers in the confidence set in \eqref{eq:confidence-set}. Planning is performed over this confidence set: for each candidate population policy, the algorithm evaluates deviations under all plausible models, forms optimistic deviation values and a pessimistic self-play value, and scores the policy by the interval Nash-gap objective in \eqref{eq:gap-estimate}. The output is an \(\eopt\)-approximate minimizer of this objective; when the true model is contained in the confidence set, the interval objective upper bounds the true Nash gap.

\paragraph{Transition support.}
The support condition in Assumption~\ref{ass:boundedness} allows zero transition probabilities and includes deterministic kernels.
Together with total-variation continuity, it implies that, for fixed \(\theta,h,x,a\), the support of \(T_h^\theta(\cdot\mid x,a,m)\) is independent of \(m\in\cB_f\).
In particular, when the kernel is deterministic, its successor may depend on \(h,x,a\) but not on \(m\).
This support stability ensures that evaluating the realizable model at a finite-population estimate of the representation cannot assign zero probability to a transition possible at the exact representation.
These assumptions do not require distinct candidate models to share a transition support.
We therefore measure transition prediction error using squared Hellinger distance, which remains finite under support mismatch.

\begin{algorithm}[t]
\caption{Offline mean-field representation learning}
\label{alg:offline-representation-learning}
\begin{algorithmic}[1]

\Statex \textbf{Input:} Offline data \(\{(X_{n,h},A_{n,h},R_{n,h},X_{n,h+1},(X_{n,h}^{(k)},A_{n,h}^{(k)})_{k=1}^K)_{h=0}^{H-1}\}_{n=1}^N\), model class \(\Theta\), tolerance \(\beta \ge 0\), optimization tolerance \(\eopt \ge 0\).

\Statex \textbf{Output:} Policy \(\widehat\eta\in\Pi\) satisfying \(\widehat{\gap}_\beta(\widehat\eta) \le \inf_{\eta\in\Pi}\widehat{\gap}_\beta(\eta)+\eopt\).

\Statex Write \(Y_{n,h}\coloneq(X_{n,h},A_{n,h})\), \(Y_{n,h}^{(k)}\coloneq(X_{n,h}^{(k)},A_{n,h}^{(k)})\).

\For{each \(\theta=(f^\theta,r^\theta,T^\theta)\in\Theta\)}
  \State Compute \(\widehat m_{n,h}^{\theta} = K^{-1}\sum_{k=1}^K f_h^\theta(Y_{n,h}^{(k)})\) and the empirical risk
  \begin{equation}
      \label{eq:empirical-risk}
      \widehat{\cR}(\theta)
      \coloneq
      \frac{1}{N} \sum_{n=1}^N \sum_{h=0}^{H-1}
      \left\{
      -\log T_h^\theta\!\lr{X_{n,h+1}\mid Y_{n,h},\widehat m_{n,h}^{\theta}}
      +
      \lr{R_{n,h}-r_h^\theta(Y_{n,h},\widehat m_{n,h}^{\theta})}^2
      \right\}.
  \end{equation}
  \Statex We use the convention \(-\log 0=+\infty\).
\EndFor
\State Create the confidence set
\begin{equation}
\label{eq:confidence-set}
    \cC(\beta)\coloneq\{\theta\in\Theta:\widehat{\cR}(\theta)\le\inf_{\theta'\in\Theta}\widehat{\cR}(\theta')+\beta\}
\end{equation}
\For{each candidate population policy \(\eta\in\Pi\)}
  \State Evaluate \(J_\theta(\pi;\eta)\) for all \(\theta\in\cC(\beta)\) and \(\pi\in\Pi\), and set
  \[
  \overline J_\beta(\pi;\eta)\coloneq\sup_{\theta\in\cC(\beta)}J_\theta(\pi;\eta), \quad \underline J_\beta(\pi;\eta)\coloneq\inf_{\theta\in\cC(\beta)}J_\theta(\pi;\eta), \text{ and}
  \]
  \begin{equation}
      \label{eq:gap-estimate}
      \widehat{\gap}_\beta(\eta)\coloneq\sup_{\pi\in\Pi}\overline J_\beta(\pi;\eta)-\underline J_\beta(\eta;\eta).
  \end{equation}
\EndFor
\State \Return any policy satisfying the output condition.
\end{algorithmic}
\end{algorithm}

\section{Provable Representation Learning in Mean-Field RL}
\label{sec:theory}

We now state the main guarantee for the offline procedure from Section~\ref{sec:offline}.
The likelihood-based empirical criterion yields a confidence set with small population prediction loss along the logged mean-field representation paths.
This loss controls logged Bellman residuals, and the coefficient \(\kappa\) transfers those residuals to value errors for a target population policy.
Optimistic deviation values and pessimistic self-play values then turn these value bounds into a Nash-gap guarantee.

\subsection{Model confidence from data}
\label{subsec:finite-confidence}
We begin with the population prediction error controlled by the empirical criterion. For a generic logged episode, draw \(\rho\sim\fp\), and conditionally on \(\rho\), write \(Y_h=(X_h,A_h)\).

The learner fits models by negative log-likelihood, while squared Hellinger distance is used to measure population prediction error in the analysis.
To understand this choice, note that a candidate model may fit every observed transition while assigning zero probability to a rare but possible transition that did not appear in the dataset.
Its forward-KL error is then infinite, while the Hellinger error remains finite and scales with the probability mass of the discrepancy.
In the analysis, the likelihood criterion controls Hellinger prediction error, and Hellinger in turn controls the total-variation error needed in the Bellman analysis.
This support issue is present even if the mean-field representation is known exactly.
To this end, define
\begin{equation}
\label{eq:population-prediction-loss}
\cL(\theta)\coloneq
\sum_{h=0}^{H-1}\EE_{\rho\sim\fp}\EE_{Y\sim\nu_h^{\star,\rho}}\bigg[2\hel^2\!\left(
T_h^\star(\cdot\mid Y,m_h^{\star,\rho}),
T_h^\theta(\cdot\mid Y,\bar m_h^{\theta,\rho})
\right)+\left(
r_h^\star(Y,m_h^{\star,\rho})
-r_h^\theta(Y,\bar m_h^{\theta,\rho})\right)^2\bigg].
\end{equation}
Under Assumption~\ref{ass:realizability}, \(\cL(\theta^\star)=0\). For \(s\ge0\), define
\[
\Theta_s\coloneq\{\theta\in\Theta:\cL(\theta)\le s\}.
\]

\begin{theorem}[Finite-class confidence event, informal constants]
\label{thm:finite-class-confidence}
Suppose Assumptions~\ref{ass:realizability}--\ref{ass:continuity} hold and \(\Theta\) is finite. There are choices of \(\beta\) and \(s_{N,K}\) with
\[
\beta,\ s_{N,K}=O\!\left(
\frac{H\bigl(1+\log(1/\tau)\bigr)\log(|\Theta|/\delta)}{N}
+
\frac{B_f^2}{K}\sum_{h=0}^{H-1}
\left[\frac{L_{T,h}^2}{\tau}+L_{r,h}^2\right]\right)
\]
such that, with probability at least \(1-\delta\), \(\theta^\star\in\cC(\beta)\) and \(\cC(\beta)\subseteq\Theta_{s_{N,K}}.\)
\end{theorem}
The \(N^{-1}\) term comes from the likelihood-based argument, while the \(K^{-1}\) term comes from evaluating candidate models at a finite-population estimate of the representation. If \(B_f,\tau^{-1},L_{T,h},L_{r,h}\) are uniformly bounded in \(h\), then
\[
s_{N,K}=O\!\left(
\frac{H\log(|\Theta|/\delta)}{N}+\frac{H}{K}\right).
\]
The exact statement and proof are given in Theorem~\ref{thm:app-finite-class-confidence} in App.~\ref{app:finite-confidence}.

\subsection{Value transfer through Bellman residuals}
\label{subsec:value-transfer}

The confidence event above controls prediction loss only at logged population policies \(\rho\sim\fp\).
Planning, however, evaluates new population policies \(\eta\) and deviations \(\pi\).
We transfer this confidence through Bellman residuals of the candidate model-induced value functions.

For \(\theta\in\Theta\) and \(\pi\in\Pi\), define suffix-indexed candidate value functions \(q_h^{\theta,\pi}\) and \(v_h^{\theta,\pi}\) by the terminal condition \(v_H^{\theta,\pi}(x,\varnothing)=0\) and the recursion in \eqref{eq:suffix-value-recursion} with \(\star\) replaced by \(\theta\).
For a true current representation \(m\in\cB_f\) and a candidate suffix \(z_{h:H-1}\), define the scalar Bellman residual \(\mathfrak b_h^{\theta,\pi}(x,a;m,z_{h:H-1})\) by
\[
\mathfrak b_h^{\theta,\pi}(x,a;m,z_{h:H-1}) \coloneqq
 r_h^\theta(x,a,z_h)-r_h^\star(x,a,m)+\sum_{x'\in\cX}
\bigl(T_h^\theta(x'\mid x,a,z_h)-T_h^\star(x'\mid x,a,m)\bigr)
 v_{h+1}^{\theta,\pi}(x',z_{h+1:H-1}).
\]
This is the Bellman residual of the candidate suffix value under the true one-step Bellman operator.

For a finite positive measure \(\mathsf M\) on tuples \((h,m,z_{h:H-1},x,a)\), write
\[
\|\mathfrak b^{\theta,\pi}\|_{2,\mathsf M}^2\coloneq\int
\mathfrak b_h^{\theta,\pi}(x,a;m,z_{h:H-1})^2
\,\mathrm d \mathsf M(h,m,z_{h:H-1},x,a).
\]
For each \(\theta\), define the logged and target lifted measures by
\[
\mathsf D_{\mathrm{log}}^\theta\coloneq
\sum_{h=0}^{H-1}\EE_{\rho\sim\fp}\left[
\delta_{(h,m_h^{\star,\rho},\bar m_{h:H-1}^{\theta,\rho})}
\otimes \nu_h^{\star,\rho}\right], \qquad
\mathsf D_{\mathrm{tar}}^{\theta,\pi;\eta}\coloneq\sum_{h=0}^{H-1}
\delta_{(h,m_h^{\star,\eta},m_{h:H-1}^{\theta,\eta})}
\otimes \nu_h^{\star,\pi;\eta}.
\]
The target measure weights residuals by the representative-agent occupancy under \(\pi\) against the population flow of \(\eta\), while the logged measure weights residuals by behavior population flows.

Define the Bellman-transfer concentrability coefficient
\[
\kappa(\pi;\eta;s)\coloneq\sup_{\theta\in\Theta_s}\frac{
\left(\displaystyle\int \mathfrak b^{\theta,\pi}\,\mathrm d\mathsf D_{\mathrm{tar}}^{\theta,\pi;\eta}\right)^2
}{\|\mathfrak b^{\theta,\pi}\|_{2,\mathsf D_{\mathrm{log}}^\theta}^2}
=\sup_{\theta\in\Theta_s}\frac{
\left(J_\theta(\pi;\eta)-J(\pi;\eta)\right)^2
}{\|\mathfrak b^{\theta,\pi}\|_{2,\mathsf D_{\mathrm{log}}^\theta}^2},
\]
with the conventions \(0/0=0\) and positive-over-zero equal to \(+\infty\). The proof of the last equality is in App.~\ref{app:bellman-value-transfer}. A finite coefficient means that every model in \(\Theta_s\) whose Bellman residual is small on logged behavior population flows must have small signed Bellman residual along the target evaluation of \(\pi\) against~\(\eta\).

Logged and target population policies may induce different Dirac masses in the population and representation coordinates, so the corresponding lifted measures can be mutually singular even when their state-action marginals overlap and the candidate models generalize correctly between the two contexts.

The suffix \(m_{h:H-1}^{\theta,\eta}\) in \(\mathsf D_{\mathrm{tar}}^{\theta,\pi;\eta}\) is essential. Proposition~\ref{prop:one-step-representation-not-sufficient} shows that a current one-step representation need not determine continuation values, so the Bellman residual must evaluate candidate continuation values along the future representation path. The logged denominator uses \(\bar m_{h:H-1}^{\theta,\rho}\), because the offline samples identify candidate features on the true logged population flow.

\begin{remark}[An \(L^2\) Bellman coverage coefficient]

Another natural Bellman-error coverage ratio is
\[
\mathscr{C}(\pi;\eta;s)\coloneq\sup_{\theta\in\Theta_s}\frac{
\|\mathfrak b^{\theta,\pi}\|_{2,\mathsf D_{\mathrm{tar}}^{\theta,\pi;\eta}}^2
}{\|\mathfrak b^{\theta,\pi}\|_{2,\mathsf D_{\mathrm{log}}^\theta}^2},
\]
with the same conventions.
This is similar to the Bellman-error coverage ratios used in offline Markov games \citep{zhang2023offline}.
Since \(\mathsf D_{\mathrm{tar}}^{\theta,\pi;\eta}\) has total mass \(H\), Cauchy-Schwarz inequality gives \(\kappa(\pi;\eta;s)\le H \mathscr{C}(\pi;\eta; s)\).
Thus \(\mathscr{C}\) can replace \(\kappa\) in all bounds at the cost of an additional factor \(\sqrt H\) in the value-transfer radius.

The coefficient \(\kappa\) measures only the transfer needed for the signed value error, whereas the coefficient \(\mathscr{C}\) controls the stronger target \(L^2\) Bellman-residual error.
\end{remark}

Let \(a_H\coloneq 1+(H-1)^2\).
App.~\ref{app:bellman-value-transfer} proves the bound \(\|\mathfrak b^{\theta,\pi}\|_{2,\mathsf D_{\mathrm{log}}^\theta}^2 \le a_H\cL(\theta),\) for all \(\theta\in\Theta,\;\pi\in\Pi\).
For \(s>0\), define \(R_s(\pi;\eta) \coloneq \sqrt{a_Hs\,\kappa(\pi;\eta;s)}.\)
Then, whenever \(\theta\in\Theta_s\), \(|J_\theta(\pi;\eta)-J(\pi;\eta)| \le R_s(\pi;\eta).\)
The appendix uses this radius to control the one-sided interval errors of the optimistic--pessimistic planner.

\subsection{Nash-gap guarantee}
\label{subsec:nash-gap-guarantee}

We now combine the confidence event with the optimistic--pessimistic planner.
For a population policy \(\eta\), define the optimistic response sub-optimality of \(\tilde\pi\) by \(\operatorname{subopt}^{+}_{\beta,\eta}(\tilde\pi) \coloneq \sup_{\pi\in\Pi}\overline J_\beta(\pi;\eta)-\overline J_\beta(\tilde\pi;\eta).\)
This term is zero when \(\tilde\pi\) is an optimistic best response to \(\eta\), and it is small when \(\tilde\pi\) is a good optimistic approximate response.

\begin{theorem}[Adaptive Nash-gap bound]
\label{thm:adaptive-nash-gap}
Assume that, for some \(s>0\), the event \(\theta^\star\in\cC(\beta)\subseteq\Theta_s\) holds.
Let \(\widehat\eta\) be an \(\eopt\)-approximate minimizer of \(\widehat{\gap}_\beta\).
Then, for every \(\eta\in\Pi\),
\[
\gap(\widehat\eta)\le\eopt+\gap(\eta)+R_s(\eta;\eta)
+\inf_{\tilde\pi\in\Pi}
\left\{R_s(\tilde\pi;\eta)+\operatorname{subopt}^{+}_{\beta,\eta}(\tilde\pi)\right\}.
\]
\end{theorem}

The proof is given in App.~\ref{app:nash-gap-proof}.
This theorem is the main offline equilibrium statement.
It says that learning succeeds when two things are true for some low-gap policy \(\eta\): the data can evaluate its self-play value, and the data can evaluate at least one good approximate response to it.
It does not require coverage of every possible unilateral deviation.

A simpler corollary is obtained by replacing the adaptive response term by a worst-case Bellman-transfer coefficient, \(\kappa_{\mathrm{uni}}(\eta;s) \coloneq \max\left\{\kappa(\eta;\eta;s),\sup_{\pi\in\Pi}\kappa(\pi;\eta;s)\right\}.\)

\begin{corollary}[Equilibrium rate]
\label{cor:equilibrium-rate}
Under the assumptions of Theorem~\ref{thm:finite-class-confidence}, with probability at least \(1-\delta\), every \(\eopt\)-approximate minimizer \(\widehat\eta\) of \(\widehat{\gap}_\beta\) satisfies
\[
\gap(\widehat\eta) \le \eopt + \inf_{\eta\in\Pi}\left\{
\gap(\eta) + \sqrt{a_Hs_{N,K}}\left(
\sqrt{\kappa(\eta;\eta;s_{N,K})}
+ \sup_{\pi\in\Pi}\sqrt{\kappa(\pi;\eta;s_{N,K})}
\right)\right\}.
\]
In particular, for any equilibrium \(\eta_\star\), the bound is at most \(\eopt+2\sqrt{a_Hs_{N,K}}\sqrt{\kappa_{\mathrm{uni}}(\eta_\star;s_{N,K})}\).
\end{corollary}

The proof is given in App.~\ref{app:nash-gap-proof}.
When \(B_f,\tau^{-1},L_{T,h},L_{r,h}\) are bounded uniformly in \(h\), the last display gives the rate
\[
\gap(\widehat\eta) \le \eopt + O\!\left(
\inf_{\eta_\star\in\Pi_\star}\sqrt{\kappa_{\mathrm{uni}}(\eta_\star;s_{N,K})}
\left[\sqrt{\frac{a_HH\log(|\Theta|/\delta)}{N}}
+ \sqrt{\frac{a_HH}{K}}\right]\right).
\]
The \(N^{-1/2}\) term is the usual offline model-selection error from observing \(N\) logged representative trajectories.
The \(K^{-1/2}\) term is specific to mean-field representation learning: each logged episode contains only a finite population sample, so the learner observes \(\widehat m_h^\theta\) rather than the population representation \(\bar m_h^{\theta,\rho}\).
Even with infinitely many episodes, finite \(K\) creates a plug-in error in the representation contexts used by the reward and transition models.
The confidence and Nash-gap bounds have no explicit dependence on the cardinality of the state space.

\section{Experiments}
\label{sec:experiments}

We design and study a routing game motivated by large-scale fulfillment and supply-chain optimization problems in the real world.
Consider a directed network of origin nodes, destination nodes, and some intermediate nodes.
Each agent is assigned an origin-destination pair and has to choose intermediate nodes to route its flow.
There may be multiple agents assigned to an origin-destination pair.
Depending on the number of agents choosing to route their flow through a particular edge in this graph, each agent's flow on that edge suffers a latency.
The reward of an agent is a function of the negative travel cost of its route under the congestion created by the entire population.
This game is naturally modeled as a one-step mean-field game since the rewards of an agent depend on the mean-field distribution as opposed to the behavior of each individual agent.

The routing game mentioned above is a one-step game, i.e., \(H=1\).
Therefore, the agent must learn only the unknown reward function, which in turn depends on the unknown low-dimensional mean-field representation, and not the transition dynamics.
Despite this, the game is rich enough to test our hypothesis that learning a low-dimensional mean-field representation is advantageous.
Algorithm~\ref{alg:offline-representation-learning} is not computationally tractable, so we isolate its central, implementable representation-learning component.
From offline rows containing a focal state--action pair, a finite \(K\)-sample population snapshot, and a scalar reward label, we fit neural reward models that differ only in how they encode the population.
We then solve the one-step game induced by each fitted model and evaluate both population-weighted reward prediction error and the exact true Nash gap of the resulting policy.

\subsection{The routing game environment}
\label{subsec:exp-routing-game}

We make the routing game environment more precise (see Figure~\ref{fig:representation_illustration}).
The directed network has five origins \(O_1,\ldots,O_5\), five destinations \(D_1,\ldots,D_5\), and two intermediate layers \(U_1,\ldots,U_4\) and \(V_1,\ldots,V_4\).
A state is an origin--destination pair, \(x=(i,j)\in\{1,\ldots,5\}^2\), making \(|\cX|=25\).
An action selects one node in each intermediate layer, \(a=(u,v)\in\{1,\ldots,4\}^2\), making \(|\cA|=16\).

The state--action pair \((x,a)=((i,j),(u,v))\) specifies the route \(O_i\rightarrow U_u\rightarrow V_v\rightarrow D_j\).
There are \(5\cdot4+4\cdot4+4\cdot5=56\) directed edges, and every route uses exactly three of them.
Let \(f^\star\colon\cX\times\cA\to\{0,1\}^{56}\) be the true feature map, where \(f_e^\star(x,a)=1\) if the route specified by \((x,a)\) uses edge \(e\), and \(f_e^\star(x,a)=0\) otherwise.
For a population state--action law \(\nu\in\Delta(\cX\times\cA)\), the true mean-field representation is \(m^\star\coloneq\EE_{(X,A)\sim\nu}[f^\star(X,A)]\).
Its \(e\)-th coordinate is the population mass routed through edge \(e\).
Thus the full population law has 400 coordinates, whereas the mean-field representation has only 56.

Each edge \(e\) has a minimum latency \(\tau_e>0\), capacity \(c_e>0\), and positive congestion coefficients \(\alpha_e,\beta_e\).
Its latency at load \(z\in[0,1]\) is
\begin{equation}
\label{eq:exp-routing-latency}
\ell_e(z)
=
\tau_e + \alpha_e\frac{z}{c_e} + \beta_e\left(\frac{z}{c_e}\right)^2.
\end{equation}
With a small fixed route offset \(\delta_{x,a}\), define
\begin{equation}
\label{eq:exp-routing-reward}
C(x,a,m)
=
\delta_{x,a}+\sum_{e=1}^{56}f_e^\star(x,a)\ell_e(m_e),
\qquad
r^\star(x,a,m)
=
1-\frac{C(x,a,m)}{C_{\max}},
\end{equation}
where \(C_{\max}=\max_{x,a,m}C(x,a,m)\), with the maximum taken over \(x\in\cX\), \(a\in\cA\), and mean-field representations \(m=\EE_{(X,A)\sim\nu}[f^\star(X,A)]\) induced by \(\nu\in\Delta(\cX\times\cA)\), so that \(r^\star \in [0,1]\).
The edge parameters and route offsets are asymmetric and fixed.
Consequently, a route that is attractive under one population can be poor under another.
In particular, the uniform routing policy is not automatically close to equilibrium, as it might first seem plausible.

\paragraph{Offline data generation.}
Suppressing the sole stage index \(h=0\), the data follow the observation model of Section~\ref{sec:offline}, except that the demand law also varies across rows.
For row \(n\in\{1,\ldots,N\}\), a context specifies an origin--destination demand law \(\mu_n\) and a behavior population policy \(\rho_n\), which induce \(\nu_n(x,a)=\mu_n(x)\rho_n(a\mid x)\).
Write \(m_n^\star\coloneq\EE_{\nu_n}[f^\star(X,A)]\) for the corresponding true representation.
Conditional on that context, we draw a representative-agent state--action pair \((X_n,A_n)\sim\nu_n\) and, independently, a population snapshot \((X_n^{(k)},A_n^{(k)})_{k=1}^K\) with \((X_n^{(k)},A_n^{(k)})\stackrel{\mathrm{i.i.d.}}{\sim}\nu_n\).
The observed reward is \(R_n=r^\star(X_n,A_n,m_n^\star)\), so the finite snapshot affects the model input but not the reward label.
Thus \(N\) counts logged representative-agent observations, whereas \(K\) controls the accuracy of each row's empirical population input without adding noise to \(R_n\).
The context generator and the precise data set construction are given in Appendix~\ref{app:exp-routing-data}.

\subsection{The reward models}
\label{subsec:exp-routing-models}

For \((x,a)=((i,j),(u,v))\), let \(z(x,a)\in\{0,1\}^{26}\) concatenate three one-hot encodings: 5 coordinates for the origin \(i\), 5 for the destination \(j\), and 16 for the action \((u,v)\); so \(z(x,a)\) has exactly three entries equal to one with the rest being zero.
For notational consistency with Section~\ref{sec:offline}, write \(f^\theta(x,a)\) for the output of the learned population encoder applied to \(z(x,a)\).
The mean of \(z\) reveals the origin and destination marginals and the loads on the 16 middle edges \(U_u\to V_v\).
It does not reveal the joint frequencies of origins and selected \(U\)-nodes, or of selected \(V\)-nodes and destinations, which determine the 40 outer-edge loads.
This makes the raw mean useful but structurally incomplete.

We compare six reward models that encode different amounts and structures of population information.
All six models share the same focal encoder, reward head, optimizer, data, and checkpoint-selection rule.

\begin{enumerate}[label=(\roman*)]
\item \emph{Single-agent.}  The model receives \((x,a)\) but no population input.

\item \emph{Monolithic raw mean.}  The model averages the raw encodings of the \(K\) population samples and then applies a nonlinear population network, \(g^\theta\!\left(\frac1K\sum_{k=1}^K z(X_n^{(k)},A_n^{(k)})\right)\).

\item \emph{Learned mean field.}  For row \(n\), the model encodes each population sample before averaging,
\[
\widehat m_n^\theta
\coloneq
\frac1K\sum_{k=1}^K
f^\theta(X_n^{(k)},A_n^{(k)}).
\]
Figure~\ref{fig:learned-mean-field-architecture} depicts the learned mean-field reward architecture: a shared encoder is applied to each population sample before mean pooling, and the pooled representation is concatenated with a separate embedding of the focal state--action pair.
Because \(z(x,a)\) uniquely identifies the route, a suitable encoder can satisfy \(f^\theta(x,a)=f^\star(x,a)\). For such \(\theta\), \(\widehat m_n^\theta\) is an empirical estimate of \(m_n^\star\).

\item \emph{Finite-\(K\) oracle.}  For each population sample, the model is given the output of the true feature map \(f^\star\), but it is trained from the empirical representation \(\frac1K\sum_{k=1}^K f^\star(X_n^{(k)},A_n^{(k)})\).
It removes representation-learning error while retaining finite-snapshot noise.

\item \emph{Infinite-population oracle.}  For each training row \(n\), the model is given the analytic representation \(m_n^\star\).
It removes both representation learning and finite-\(K\) input noise, but it still learns the reward function from finitely many offline rows.

\item \emph{Full population law.} For row \(n\), let \(\widehat\nu_n\in\Delta(\cX\times\cA)\) be the empirical population law, with
\[
\widehat\nu_n(x,a)
=
\frac1K\sum_{k=1}^K
\mathbf 1\{(X_n^{(k)},A_n^{(k)})=(x,a)\}.
\]
That is, \(\widehat\nu_n\) is an empirical histogram over the 400 state--action categories.
This model feeds \(\widehat\nu_n\) to a parameter-matched population encoder.
It is therefore a learned full-information baseline.
\end{enumerate}

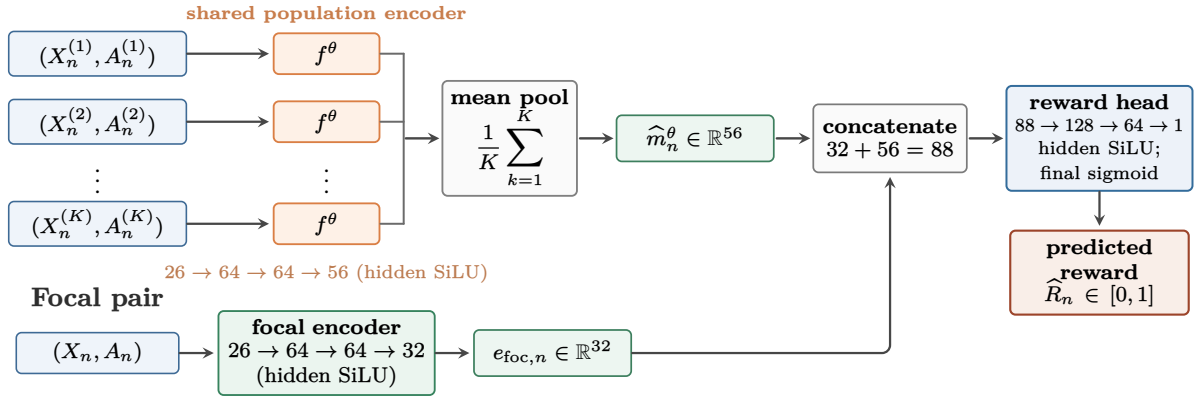
\begin{figure}[H]
\centering
\begin{tikzpicture}[
  x=1cm,
  y=1cm,
  font=\footnotesize,
  input box/.style={draw=RoyalBlue!65!black, fill=RoyalBlue!5,
    rounded corners=2pt, line width=0.6pt, minimum width=2.35cm,
    minimum height=0.54cm, align=center, inner sep=3pt},
  population encoder/.style={draw=Orange!88!black, fill=Orange!9,
    rounded corners=2pt, line width=0.7pt, minimum width=1.40cm,
    minimum height=0.54cm, align=center, inner sep=3pt},
  operation/.style={draw=black!50, fill=black!2, rounded corners=2pt,
    line width=0.7pt, minimum width=1.60cm, minimum height=0.86cm,
    align=center, inner sep=3pt},
  focal encoder/.style={draw=ForestGreen!70!black, fill=ForestGreen!7,
    rounded corners=2pt, line width=0.7pt, text width=2.65cm,
    minimum height=0.78cm, align=center, inner sep=3pt},
  embedding/.style={draw=ForestGreen!65!black, fill=ForestGreen!6,
    rounded corners=2pt, line width=0.65pt, minimum width=2.08cm,
    minimum height=0.61cm, align=center, inner sep=3pt},
  merge box/.style={draw=black!52, fill=black!2, rounded corners=2pt,
    line width=0.65pt, text width=1.82cm, minimum height=0.90cm,
    align=center, inner sep=3pt},
  head box/.style={draw=RoyalBlue!70!black, fill=RoyalBlue!6,
    rounded corners=2pt, line width=0.75pt, text width=2.25cm,
    minimum height=0.96cm, align=center, inner sep=3pt},
  prediction/.style={draw=BrickRed!80!black, fill=BrickRed!6,
    rounded corners=2pt, line width=0.8pt, text width=2.15cm,
    minimum height=0.86cm, align=center, inner sep=3pt},
  flow/.style={->, >=stealth, draw=black!72, line width=0.85pt,
    line cap=round, line join=round, shorten >=1.4pt},
  stage flow/.style={->, >=stealth, draw=black!72, line width=0.90pt,
    line cap=round, line join=round, shorten <=0.8pt, shorten >=0.8pt},
  bus/.style={draw=black!58, line width=0.75pt, line cap=round}
]
\path[use as bounding box] (0,0) rectangle (16.00,5.58);

\node[font=\bfseries, text=black!82] at (1.30,5.30)
  {Population samples};
\node[input box] (pop1) at (1.30,4.30)
  {\((X_n^{(1)},A_n^{(1)})\)};
\node[input box] (pop2) at (1.30,3.40)
  {\((X_n^{(2)},A_n^{(2)})\)};
\node at (1.30,2.72) {\(\vdots\)};
\node[input box] (popK) at (1.30,2.05)
  {\((X_n^{(K)},A_n^{(K)})\)};

\node[inner xsep=4pt, text=Orange!78!black,
  font=\bfseries\scriptsize] at (4.33,4.82)
  {shared population encoder};
\node[population encoder] (enc1) at (4.33,4.30) {\(f^\theta\)};
\node[population encoder] (enc2) at (4.33,3.40) {\(f^\theta\)};
\node at (4.33,2.72) {\(\vdots\)};
\node[population encoder] (encK) at (4.33,2.05) {\(f^\theta\)};
\node[align=center, text=Orange!72!black, font=\scriptsize]
  at (4.33,1.40)
  {\(26\to64\to64\to56\) (hidden SiLU)};

\draw[flow] (pop1.east) -- (enc1.west);
\draw[flow] (pop2.east) -- (enc2.west);
\draw[flow] (popK.east) -- (encK.west);
\draw[bus] (enc1.east) -- (5.35,4.30);
\draw[bus] (enc2.east) -- (5.35,3.40);
\draw[bus] (encK.east) -- (5.35,2.05);
\draw[bus] (5.35,2.05) -- (5.35,4.30);
\coordinate (population-bus) at (5.35,3.18);

\node[operation, right=0.50cm of population-bus] (pool)
  {\textbf{mean pool}\\[-1pt]
   \(\displaystyle\frac1K\sum_{k=1}^{K}\)};
\draw[stage flow] (population-bus) -- (pool.west);

\node[embedding, right=0.50cm of pool] (meanemb)
  {\(\widehat m_n^\theta\in\RR^{56}\)};
\draw[stage flow] (pool.east) -- (meanemb.west);

\node[font=\bfseries, text=black!82] at (1.30,1.05)
  {Focal pair};
\node[input box, minimum width=2.15cm] (focal) at (1.30,0.34)
  {\((X_n,A_n)\)};
\node[focal encoder] (focenc) at (4.33,0.34)
  {\textbf{focal encoder}\\[-1pt]
   \(26\to64\to64\to32\) (hidden SiLU)};
\node[embedding] (focemb) at (7.32,0.34)
  {\(e_{\mathrm{foc},n}\in\RR^{32}\)};
\draw[flow] (focal.east) -- (focenc.west);
\draw[flow] (focenc.east) -- (focemb.west);

\node[merge box, right=0.50cm of meanemb] (merge)
  {\textbf{concatenate}\\[-1pt]
   \(32+56=88\)};
\draw[stage flow] (meanemb.east) -- (merge.west);
\draw[flow, rounded corners=3pt]
  (focemb.east) -| (merge.south);

\node[head box, right=0.50cm of merge] (head)
  {\textbf{reward head}\\[-1pt]
   \scalebox{0.86}{\(88\to128\to64\to1\)}\\[-1pt]
   {\scriptsize hidden SiLU; final sigmoid}};
\draw[stage flow] (merge.east) -- (head.west);

\node[prediction, below=0.50cm of head] (pred)
  {\textbf{predicted reward}\\[-1pt]
   \(\widehat R_n\in[0,1]\)};
\draw[flow] (head.south) -- (pred.north);
\end{tikzpicture}
\caption{Architecture of the learned mean-field reward model. Each population sample is mapped by the shared encoder \(f^\theta\), and the \(K\) encoded vectors are averaged to form \(\widehat m_n^\theta\). In parallel, a separate focal encoder embeds the focal state--action pair. The focal and population embeddings are concatenated and passed through the reward head to predict the scalar reward.}
\label{fig:learned-mean-field-architecture}
\end{figure}

\subsection{Evaluation metrics}
\label{subsec:exp-routing-evaluation}

Final evaluation removes avoidable population-sampling noise.
For each evaluation context \(c\), let \(\nu_c\) be its induced population law and let \(m_c^\star\coloneq\sum_{x,a}\nu_c(x,a)f^\star(x,a)\) be its true representation.
Every model receives the exact conditional population input appropriate to its architecture: the learned model receives \(\sum_{x,a}\nu_c(x,a)f^\theta(x,a)\), the monolithic model receives \(g^\theta(\sum_{x,a}\nu_c(x,a)z(x,a))\), the two oracles receive \(m_c^\star\), the full-law model receives \(\nu_c\), and the single-agent model receives no population input.
We then enumerate all 400 state--action categories exactly.
Writing \(r_c^\star(x,a)=r^\star(x,a,m_c^\star)\), the primary reward metric is
\begin{equation}
\label{eq:exp-routing-rmse}
\operatorname{RMSE}_{\mathrm{pop}}
=
\left[
\EE_c\sum_{x,a}\nu_c(x,a)
\bigl(\widehat r_c(x,a)-r_c^\star(x,a)\bigr)^2
\right]^{1/2}.
\end{equation}
For each context, the sum over all 400 state--action pairs is evaluated exactly. The expectation over contexts is approximated using fixed randomized quasi--Monte Carlo context samples shared across models.
Further details are given in Appendix~\ref{app:exp-routing-evaluation}.

For policy evaluation, fix a target demand law \(\mu\) and define \(m^{\star,\eta} \coloneq \sum_{x,a}\mu(x) \eta(a\mid x)f^\star(x,a)\).
The exact one-step Nash gap is
\begin{equation}
\label{eq:exp-routing-gap}
\gap_\mu(\eta)
=
\sum_x\mu(x)
\left[
\max_a r^\star(x,a,m^{\star,\eta})
-
\sum_a\eta(a\mid x)r^\star(x,a,m^{\star,\eta})
\right].
\end{equation}
This is the expected gain from a unilateral best response while the population flow is held fixed.
Multiplying by \(C_{\max}\) gives the equivalent mean excess route cost.

For each fitted reward model, we use entropic mirror descent, entropic Mirror--Prox, and Adam optimization of a smoothed Nash-gap objective to generate candidate policies for the one-step game obtained by replacing the true reward function with the fitted model.
We select the candidate with the smallest Nash residual under the fitted reward model and use the true reward only afterward to evaluate that policy via~\eqref{eq:exp-routing-gap}.
We also report the uniform policy and an exact-reward control obtained by applying the optimization methods described above directly to the true reward.
The exact-reward control attains an exact true Nash gap of \(9.58\times10^{-6}\).
The optimization methods and control details are given in Appendix~\ref{app:exp-routing-planning}.
All primary curves report means and 95\% Student-\(t\) half-widths over the five paired training seeds after the metric for each seed is averaged over the shared quasi--Monte Carlo context samples.

\subsection{Results}
\label{subsec:exp-results}

\begin{figure}[t]
\centering
\includegraphics[width=\linewidth]{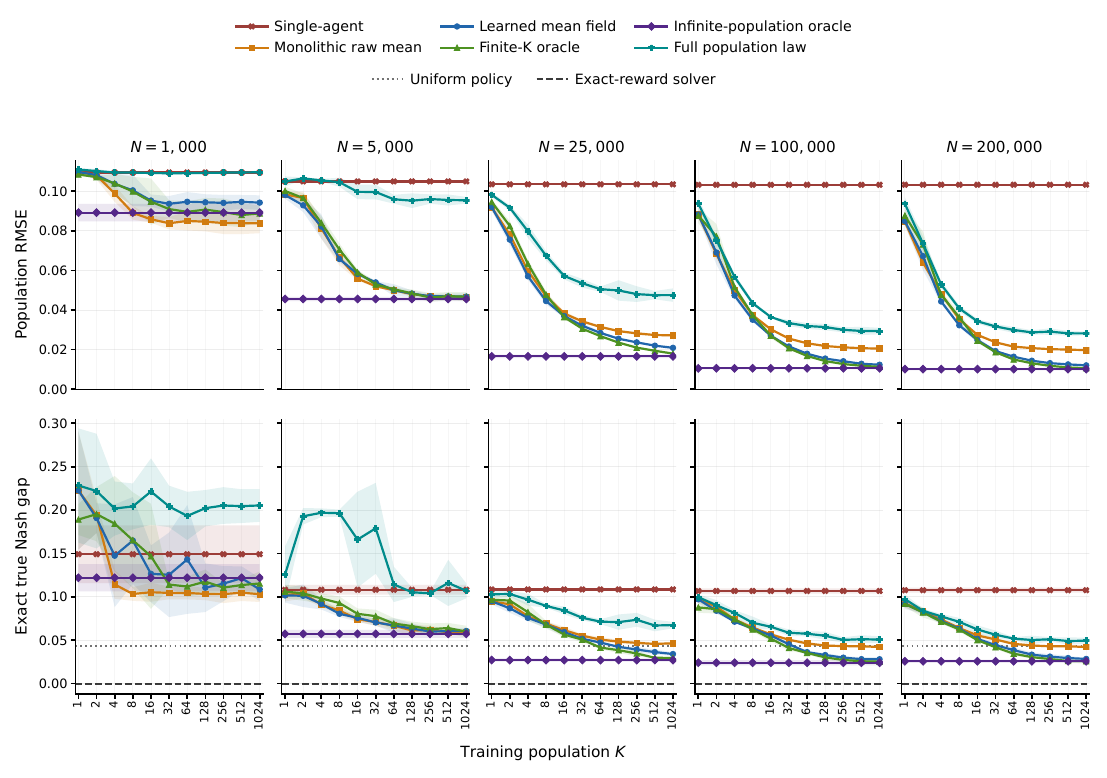}
\caption{Population reward error (upper row) and exact true Nash gap of policies obtained by solving the fitted games (lower row), as functions of the number \(N\) of offline rows and the number \(K\) of population observations in each row.
Lower is better.
For each fitted reward model, we select the policy that appears closest to equilibrium according to that model and then evaluate its Nash gap using the true reward.
The horizontal lines in the lower row show the Nash gaps of the uniform policy and the exact-reward control.
Lines connect observed means only; intervals are 95\% Student-\(t\) intervals over five paired training seeds.}
\label{fig:exp-routing-landscapes}
\label{fig:exp-routing-rmse-landscape}
\label{fig:exp-routing-gap-landscape}
\end{figure}

Figure~\ref{fig:exp-routing-landscapes} shows both statistical resources over the same \(N\)-by-\(K\) grid.
The single-agent and infinite-population-oracle curves are constant in \(K\) by construction.
At \(K=1\), the operation-order advantage of encode-then-average disappears because there is only one population sample to pool. The data exhibit precisely this sanity check.
As \(K\) grows, the models separate according to the information they can recover, and the finite-\(K\) oracle approaches the infinite-population oracle.

Increasing \(N\) exposes the same hierarchy.
At the high-resource corner \((N,K)=(200000,1024)\), the five models forming the population-information hierarchy are ordered, up to the expected convergence of the two oracles, by both population RMSE and Nash gap as
\[
\text{infinite-population oracle}
\approx
\text{finite-\(K\) oracle}
<
\text{learned mean field}
<
\text{monolithic raw mean}
<
\text{single-agent}.
\]

The full-population-law baseline is especially informative.
It contains every population coordinate, but must learn from a noisy 400-dimensional histogram without the edge-load inductive bias.
On average, it performs worse than the learned mean-field representation model.
This comparison shows that access to the full law is not by itself an effective finite-data representation under a matched parameter and optimization budget.
The full law is sufficient in principle, and a different architecture or larger training budget could change its performance.

Appendix~\ref{app:experiments} gives the exact evaluation, planning controls, and statistical conventions.

\section{Related Work}
Our work is most related to the following lines of research.

\paragraph{Mean-field RL.} Finite-state and discrete-time mean-field games provide a natural bridge to reinforcement learning \citep{gomes2010discrete,saldi2018markov}.
Existing mean-field RL algorithms adapt tools from model-free RL and mirror-descent, often with neural network function approximation, to mean-field equilibrium computation \citep{yang2018meanfieldmarl,guo2019learningMFG,lauriere2022scalable}.
Other work establishes statistical guarantees for mean-field RL under prescribed mean-field embeddings or model classes \citep{wang2020breaking,pasztor2023efficient,huang2024modelbasedmfg,huang2024statisticalmfrl}.
Our setting differs by learning from offline data and treating the low-dimensional population representation itself as unknown.

\paragraph{Offline reinforcement learning.} In single-agent offline RL, distribution shift motivates pessimistic or conservative principles, including conservative value estimation, Bellman-consistent pessimism, and model-based pessimism under partial coverage \citep{jin2021pessimism,xie2021bellman,uehara2022pessimisticmodel,jiang2025offline}.
In offline Markov games, equilibrium learning requires controlling both the candidate policy value and the values of unilateral deviations or best responses, leading to unilateral, strategy-wise, and other game-specific coverage notions \citep{cui2022when,cui2022strategywise,zhang2023offline,zhan2024exploiting}.
Complementary offline mean-field precedents are \citep{chen2021pessimism}, which studies offline mean-field RL under invariance structure, and \citep{brunnbauer2025scalable}, which proposes scalable offline mirror descent for MFGs.
However, representation learning for the mean-field is not their focus.

\paragraph{Coverage and distribution shift in offline mean-field learning.}
Prior work on coverage in mean-field RL mainly considers algorithms that can collect new data during learning, either through online interaction or through access to a simulator. \citet{anahtarci2023learningmfg} assume a full-support state-sampling distribution and a behavior policy whose action probabilities are bounded away from zero. \citet{xie2021learningwhileplaying} bound density ratios between equilibrium and optimal-response visitation laws. \citet{mao2022cloudmfg} use bounds that also involve the current policy's visitation law. Thus, these assumptions do not address learning from a fixed offline data set. In single-agent offline RL, coverage of an optimal or comparator policy can be sufficient \citep{jin2021pessimism,uehara2022pessimisticmodel}. In offline Markov games, this coverage is generally insufficient because the Nash gap also depends on unilateral deviations. This fact motivates unilateral coverage \citep{cui2022when,cui2022strategywise}. \citet{zhang2023offline} weaken this requirement by trading the coverage of an approximate response against its response suboptimality. Our setting has an additional source of distribution shift. A change in the population policy changes the population flow and the full representation path. A unilateral deviation changes only the representative agent's occupancy and leaves the population path fixed. We control both shifts with a Bellman-transfer condition. We require this condition only for a low-gap population policy and one approximate response.

\paragraph{Representation learning.} Representation learning is central in low-rank and latent-structure RL, where compact features can make exploration, planning, and offline policy optimization statistically tractable \citep{misra2020kinematic,agarwal2020flambe,uehara2022representation,zhang2023relex}. Related ideas also appear in low-rank Markov games \citep{ni2023representationgames}.
Our model-learning component is closest in spirit to representation learning for low-rank MDPs. In mean-field RL, however, the representation is tied to the population policy being evaluated: changing the policy changes the population flow and therefore the sequence of representations over time. Any unilateral deviation must then be evaluated while holding this policy-induced sequence fixed.

\section{Conclusion and Future Work}
\label{sec:conclusion}
We studied offline mean-field RL when the population affects rewards and transitions only through an unknown low-dimensional representation.
We formalized this representation-learning problem, showed that the resulting equilibrium objective is well defined, and proposed a model-based algorithm.
Under realizability and adaptive Bellman-transfer coverage, we proved a finite-sample Nash-gap guarantee.
In the one-step routing benchmark, learned mean-field representations improve both population reward prediction and the equilibrium quality of policies obtained by planning in the fitted game.
The learned model performs better than a structurally incomplete raw mean, approaches representation oracles as resources grow, and outperforms a parameter-matched model of the complete population law in the high-resource regime.

Several directions remain open, which we have not addressed.
First, lower bound analysis is needed to identify the weakest coverage conditions under which offline equilibrium learning is possible.
Second, our algorithm is statistically motivated but computationally inefficient; designing scalable algorithms with provable comparable guarantees remains  an important problem.
Third, an online version of mean-field representation learning would be valuable.
Such an extension would need new exploration criteria over population flows and unilateral deviations.

\section*{Acknowledgments}
YE is partially supported by the Israeli Science Foundation (ISF) grant no 4032/25. Aditya Makkar is supported by NYU startup funds. 

\bibliographystyle{plainnat}
\bibliography{arxiv}

@article{huang2006large,
  title={Large Population Stochastic Dynamic Games: Closed-Loop McKean--Vlasov Systems and the Nash Certainty Equivalence Principle},
  author={Huang, Minyi and Malham{\'e}, Roland P. and Caines, Peter E.},
  journal={Communications in Information and Systems},
  volume={6},
  number={3},
  pages={221--252},
  year={2006},
  doi={10.4310/CIS.2006.v6.n3.a5}
}

@article{lasry2007mean,
  title={Mean Field Games},
  author={Lasry, Jean-Michel and Lions, Pierre-Louis},
  journal={Japanese Journal of Mathematics},
  volume={2},
  number={1},
  pages={229--260},
  year={2007},
  doi={10.1007/s11537-007-0657-8}
}

@book{carmona2018probabilistic,
  title={Probabilistic Theory of Mean Field Games with Applications I: Mean Field FBSDEs, Control, and Games},
  author={Carmona, Ren{\'e} and Delarue, Fran{\c c}ois},
  series={Probability Theory and Stochastic Modelling},
  volume={83},
  publisher={Springer},
  year={2018},
  doi={10.1007/978-3-319-58920-6}
}

@article{gomes2010discrete,
  title={Discrete Time, Finite State Space Mean Field Games},
  author={Gomes, Diogo A. and Mohr, Joana and Souza, Rafael Rig{\~a}o},
  journal={Journal de Math{\'e}matiques Pures et Appliqu{\'e}es},
  volume={93},
  number={3},
  pages={308--328},
  year={2010},
  doi={10.1016/j.matpur.2009.10.010}
}

@article{saldi2018markov,
  title={Markov--Nash Equilibria in Mean-Field Games with Discounted Cost},
  author={Saldi, Naci and Ba{\c s}ar, Tamer and Raginsky, Maxim},
  journal={SIAM Journal on Control and Optimization},
  volume={56},
  number={6},
  pages={4256--4287},
  year={2018},
  doi={10.1137/17M1112583}
}

@article{lauriere2022learning,
  title={Learning in Mean Field Games: A Survey},
  author={Lauri{\`e}re, Mathieu and Perrin, Sarah and P{\'e}rolat, Julien and Girgin, Sertan and Muller, Paul and {\'E}lie, Romuald and Geist, Matthieu and Pietquin, Olivier},
  journal={arXiv preprint arXiv:2205.12944},
  year={2022},
  eprint={2205.12944},
  archivePrefix={arXiv},
  primaryClass={cs.LG}
}

@article{guo2019learningMFG,
  title={Learning Mean-Field Games},
  author={Guo, Xin and Hu, Anran and Xu, Renyuan and Zhang, Junzi},
  journal={Advances in Neural Information Processing Systems},
  volume={32},
  pages={4966--4976},
  year={2019}
}

@article{kizilkale2013meanfield,
  title={Mean Field Stochastic Adaptive Control},
  author={Kizilkale, Arman C. and Caines, Peter E.},
  journal={IEEE Transactions on Automatic Control},
  volume={58},
  number={4},
  pages={905--920},
  year={2013},
  doi={10.1109/TAC.2012.2228032}
}

@article{yin2014learningMFG,
  title={Learning in Mean-Field Games},
  author={Yin, Huibing and Mehta, Prashant G. and Meyn, Sean P. and Shanbhag, Uday V.},
  journal={IEEE Transactions on Automatic Control},
  volume={59},
  number={3},
  pages={629--644},
  year={2014},
  doi={10.1109/TAC.2013.2287733}
}

@inproceedings{yang2018meanfieldmarl,
  title={Mean Field Multi-Agent Reinforcement Learning},
  author={Yang, Yaodong and Luo, Rui and Li, Minne and Zhou, Ming and Zhang, Weinan and Wang, Jun},
  booktitle={Proceedings of the 35th International Conference on Machine Learning},
  series={Proceedings of Machine Learning Research},
  volume={80},
  pages={5571--5580},
  publisher={PMLR},
  year={2018}
}

@inproceedings{lauriere2022scalable,
  title={Scalable Deep Reinforcement Learning Algorithms for Mean Field Games},
  author={Lauri{\`e}re, Mathieu and Perrin, Sarah and Girgin, Sertan and Muller, Paul and Jain, Ayush and Cabannes, Th{\'e}ophile and Piliouras, Georgios and P{\'e}rolat, Julien and {\'E}lie, Romuald and Pietquin, Olivier and Geist, Matthieu},
  booktitle={Proceedings of the 39th International Conference on Machine Learning},
  series={Proceedings of Machine Learning Research},
  volume={162},
  pages={12078--12095},
  publisher={PMLR},
  year={2022}
}

@inproceedings{jin2021pessimism,
  title={Is Pessimism Provably Efficient for Offline {RL}?},
  author={Jin, Ying and Yang, Zhuoran and Wang, Zhaoran},
  booktitle={Proceedings of the 38th International Conference on Machine Learning},
  series={Proceedings of Machine Learning Research},
  volume={139},
  pages={5084--5096},
  publisher={PMLR},
  year={2021}
}

@inproceedings{xie2021bellman,
  title={Bellman-Consistent Pessimism for Offline Reinforcement Learning},
  author={Xie, Tengyang and Cheng, Ching-An and Jiang, Nan and Mineiro, Paul and Agarwal, Alekh},
  booktitle={Advances in Neural Information Processing Systems},
  volume={34},
  pages={6683--6694},
  year={2021}
}

@article{jiang2025offline,
  title={Offline Reinforcement Learning in Large State Spaces: Algorithms and Guarantees},
  author={Jiang, Nan and Xie, Tengyang},
  journal={Statistical Science},
  volume={40},
  number={4},
  pages={570--596},
  year={2025},
  doi={10.1214/25-STS1000},
  url={https://doi.org/10.1214/25-STS1000}
}

@inproceedings{cui2022when,
  title={When Are Offline Two-Player Zero-Sum Markov Games Solvable?},
  author={Cui, Qiwen and Du, Simon S.},
  booktitle={Advances in Neural Information Processing Systems},
  volume={35},
  pages={25779--25791},
  year={2022}
}

@inproceedings{cui2022strategywise,
  title={Provably Efficient Offline Multi-Agent Reinforcement Learning via Strategy-Wise Bonus},
  author={Cui, Qiwen and Du, Simon S.},
  booktitle={Advances in Neural Information Processing Systems},
  volume={35},
  pages={11739--11751},
  year={2022}
}

@inproceedings{zhang2023offline,
  title={Offline Learning in {M}arkov Games with General Function Approximation},
  author={Zhang, Yuheng and Bai, Yu and Jiang, Nan},
  booktitle={Proceedings of the 40th International Conference on Machine Learning},
  series={Proceedings of Machine Learning Research},
  volume={202},
  pages={40804--40829},
  publisher={PMLR},
  year={2023}
}

@inproceedings{chen2021pessimism,
  title={Pessimism Meets Invariance: Provably Efficient Offline Mean-Field Multi-Agent RL},
  author={Chen, Minshuo and Li, Yan and Wang, Ethan and Yang, Zhuoran and Wang, Zhaoran and Zhao, Tuo},
  booktitle={Advances in Neural Information Processing Systems},
  volume={34},
  pages={17913--17926},
  year={2021}
}

@inproceedings{brunnbauer2025scalable,
  title={Scalable Offline Reinforcement Learning for Mean Field Games},
  author={Brunnbauer, Axel and Lemmel, Julian and Babaiee, Zahra and Neubauer, Sophie A. and Grosu, Radu},
  booktitle={Proceedings of the 24th International Conference on Autonomous Agents and Multiagent Systems},
  pages={408--417},
  publisher={International Foundation for Autonomous Agents and Multiagent Systems},
  year={2025},
  doi={10.5555/3709347.3743555}
}

@inproceedings{misra2020kinematic,
  title={Kinematic State Abstraction and Provably Efficient Rich-Observation Reinforcement Learning},
  author={Misra, Dipendra and Henaff, Mikael and Krishnamurthy, Akshay and Langford, John},
  booktitle={Proceedings of the 37th International Conference on Machine Learning},
  series={Proceedings of Machine Learning Research},
  volume={119},
  pages={6961--6971},
  publisher={PMLR},
  year={2020}
}

@inproceedings{agarwal2020flambe,
  title={{FLAMBE}: Structural Complexity and Representation Learning of Low Rank {MDP}s},
  author={Agarwal, Alekh and Kakade, Sham M. and Krishnamurthy, Akshay and Sun, Wen},
  booktitle={Advances in Neural Information Processing Systems},
  volume={33},
  pages={20095--20107},
  year={2020}
}

@inproceedings{uehara2022representation,
  title={Representation Learning for Online and Offline RL in Low-Rank {MDP}s},
  author={Uehara, Masatoshi and Zhang, Xuezhou and Sun, Wen},
  booktitle={International Conference on Learning Representations},
  year={2022}
}

@inproceedings{zhang2023relex,
  title={Provably Efficient Representation Selection in Low-Rank Markov Decision Processes: From Online to Offline RL},
  author={Zhang, Weitong and He, Jiafan and Zhou, Dongruo and Gu, Quanquan and Zhang, Amy},
  booktitle={Proceedings of the Thirty-Ninth Conference on Uncertainty in Artificial Intelligence},
  series={Proceedings of Machine Learning Research},
  volume={216},
  pages={2488--2497},
  publisher={PMLR},
  year={2023}
}

@inproceedings{zhan2024exploiting,
  title={Exploiting Structure in Offline Multi-Agent {RL}: The Benefits of Low Interaction Rank},
  author={Zhan, Wenhao and Fujimoto, Scott and Zhu, Zheqing and Lee, Jason D. and Jiang, Daniel R. and Efroni, Yonathan},
  booktitle={International Conference on Learning Representations},
  year={2025}
}

@inproceedings{huang2024statisticalmfrl,
  title={On the Statistical Efficiency of Mean-Field Reinforcement Learning with General Function Approximation},
  author={Huang, Jiawei and Yardim, Batuhan and He, Niao},
  booktitle={Proceedings of The 27th International Conference on Artificial Intelligence and Statistics},
  series={Proceedings of Machine Learning Research},
  volume={238},
  pages={289--297},
  publisher={PMLR},
  year={2024}
}

@inproceedings{huang2024modelbasedmfg,
  title={Model-Based {RL} for Mean-Field Games is not Statistically Harder than Single-Agent {RL}},
  author={Huang, Jiawei and He, Niao and Krause, Andreas},
  booktitle={Proceedings of the 41st International Conference on Machine Learning},
  series={Proceedings of Machine Learning Research},
  volume={235},
  pages={19816--19870},
  publisher={PMLR},
  year={2024}
}

@article{pasztor2023efficient,
  title={Efficient Model-Based Multi-Agent Mean-Field Reinforcement Learning},
  author={P{\'a}sztor, Barna and Krause, Andreas and Bogunovic, Ilija},
  journal={Transactions on Machine Learning Research},
  year={2023}
}

@inproceedings{wang2020breaking,
  title={Breaking the Curse of Many Agents: Provable Mean Embedding Q-Iteration for Mean-Field Reinforcement Learning},
  author={Wang, Lingxiao and Yang, Zhuoran and Wang, Zhaoran},
  booktitle={Proceedings of the 37th International Conference on Machine Learning},
  series={Proceedings of Machine Learning Research},
  volume={119},
  pages={10092--10103},
  publisher={PMLR},
  year={2020}
}

@inproceedings{subramanian2019reinforcement,
  title={Reinforcement Learning in Stationary Mean-field Games},
  author={Subramanian, Jayakumar and Mahajan, Aditya},
  booktitle={Proceedings of the 18th International Conference on Autonomous Agents and MultiAgent Systems},
  pages={251--259},
  publisher={International Foundation for Autonomous Agents and Multiagent Systems},
  year={2019}
}

@inproceedings{xie2021learningwhileplaying,
  title={Learning While Playing in Mean-Field Games: Convergence and Optimality},
  author={Xie, Qiaomin and Yang, Zhuoran and Wang, Zhaoran and Minca, Andreea},
  booktitle={Proceedings of the 38th International Conference on Machine Learning},
  series={Proceedings of Machine Learning Research},
  volume={139},
  pages={11436--11447},
  publisher={PMLR},
  year={2021}
}

@article{anahtarci2023learningmfg,
  title={Learning Mean-Field Games with Discounted and Average Costs},
  author={Anahtarci, Berkay and Kariksiz, Can Deha and Saldi, Naci},
  journal={Journal of Machine Learning Research},
  volume={24},
  number={17},
  pages={1--59},
  year={2023},
  url={https://www.jmlr.org/papers/v24/21-0505.html}
}

@inproceedings{mao2022cloudmfg,
  title={A Mean-Field Game Approach to Cloud Resource Management with Function Approximation},
  author={Mao, Weichao and Qiu, Haoran and Wang, Chen and Franke, Hubertus and Kalbarczyk, Zbigniew and Iyer, Ravishankar K. and Ba{\c{s}}ar, Tamer},
  booktitle={Advances in Neural Information Processing Systems},
  volume={35},
  pages={36243--36258},
  year={2022},
  url={https://proceedings.neurips.cc/paper_files/paper/2022/hash/eb3c8135137c8a60425a0320869ad87e-Abstract-Conference.html}
}

@inproceedings{uehara2022pessimisticmodel,
  title={Pessimistic Model-Based Offline Reinforcement Learning under Partial Coverage},
  author={Uehara, Masatoshi and Sun, Wen},
  booktitle={International Conference on Learning Representations},
  year={2022}
}

@inproceedings{ni2023representationgames,
  title={Representation Learning for Low-rank General-sum Markov Games},
  author={Ni, Chengzhuo and Song, Yuda and Zhang, Xuezhou and Ding, Zihan and Jin, Chi and Wang, Mengdi},
  booktitle={International Conference on Learning Representations},
  year={2023}
}

\appendix

\newpage

\section*{Notation Table}

\paragraph{Conventions.}
A superscript \(\star\) denotes the true model, whereas \(\theta=(f^\theta,r^\theta,T^\theta)\in\Theta\) denotes a candidate model.
In \((\pi;\eta)\), \(\pi\) is the representative-agent policy and \(\eta\) is the population policy; \(\rho\sim\fp\) denotes a logged behavior population policy.
For a sequence \(u=(u_0,\ldots,u_{H-1})\), write \(u_{h:H-1}=(u_h,\ldots,u_{H-1})\).

\begingroup
\small
\renewcommand{\arraystretch}{1.14}
\setlength{\tabcolsep}{4pt}
\begin{center}
\begin{tabular}{@{}p{0.30\linewidth}p{0.66\linewidth}@{}}
\toprule
Symbol & Meaning \\
\midrule
\multicolumn{2}{@{}l}{\textit{Flows and representations}} \\
\midrule
\(\mu_h^{\star,\eta},\nu_h^{\star,\eta}\)
& True population state and state-action laws induced by~\(\eta\), with
\(\nu_h^{\star,\eta}=\mu_h^{\star,\eta}\otimes\eta_h\). \\
\(\mu_h^{\star,\pi;\eta},\nu_h^{\star,\pi;\eta}\)
& State and state-action laws of a representative agent using \(\pi\) against
the fixed population flow induced by \(\eta\). \\
\(m_h^{\star,\eta}\)
& True representation
\(\sum_{x,a}\nu_h^{\star,\eta}(x,a)f_h^\star(x,a)\). \\
\(m_h^{\theta,\eta}\)
& Candidate model's self-consistent representation under \(\eta\), used for
planning in model \(\theta\). \\
\(\bar m_h^{\theta,\eta}\)
& Candidate feature average on the true population flow,
\(\sum_{x,a}\nu_h^{\star,\eta}(x,a)f_h^\theta(x,a)\), used in the population
prediction loss. \\
\(\widehat m_{n,h}^{\theta}\)
& Empirical candidate representation formed from the \(K\) population samples
in episode \(n\) at stage \(h\). \\
\(\PP_\star^{\pi;\eta},\PP_\theta^{\pi;\eta};
  J(\pi;\eta),J_\theta(\pi;\eta)\)
& True and candidate representative-agent path laws and their corresponding
finite-horizon values. \\
\midrule
\multicolumn{2}{@{}l}{\textit{Data and confidence}} \\
\midrule

\(N,K;\ \widehat{\cR}(\theta),\cC(\beta)\)
& Number of logged episodes and population samples per episode-stage; empirical
risk and its \(\beta\)-near-minimizer confidence set. \\
\(\cL(\theta),\Theta_s,s_{N,K}\)
& Population prediction loss,
\(\Theta_s=\{\theta\in\Theta:\cL(\theta)\le s\}\), and the confidence radius
provided by the finite-sample theorem. \\
\(B_f,\cB_f,\tau,L_{r,h},L_{T,h}\)
& Feature bound, representation ball, positive transition-probability floor,
and reward/transition Lipschitz constants. \\
\midrule
\multicolumn{2}{@{}l}{\textit{Bellman transfer and interval quantities}} \\
\midrule
\(q_h^{\theta,\pi},v_h^{\theta,\pi};
  \mathfrak b_h^{\theta,\pi}\)
& Candidate suffix-indexed action-value and value functions, and their scalar
Bellman residual under the true one-step Bellman operator. \\
\(\mathsf D_{\mathrm{log}}^\theta,
  \mathsf D_{\mathrm{tar}}^{\theta,\pi;\eta}\)
& Lifted logged and target measures over stages, representation contexts, and
state-action occupancies. \\
\(\kappa(\pi;\eta;s),\mathscr C(\pi;\eta;s)\)
& Scalar Bellman-transfer coefficient and the stronger target-\(L^2\)
Bellman-coverage coefficient. \\
\(a_H,\operatorname{rad}_{\mathrm{Bell},s},R_s\)
& \(a_H=1+(H-1)^2\), the exact Bellman-transfer radius, and
\(R_s(\pi;\eta)=\sqrt{a_Hs\,\kappa(\pi;\eta;s)}\). \\
\(\overline J_{\cC},\underline J_{\cC},
  \widehat\gap_{\cC}\)
& Optimistic and pessimistic values and the induced interval gap for a generic
plausible set \(\cC\subseteq\Theta\). \\
\(r_{\cC}^{\pm},
  \operatorname{subopt}^{+}_{\cC,\eta},\eopt\)
& One-sided interval errors, optimistic response suboptimality, and optimization
tolerance; set \(\cC=\cC(\beta)\) for the main-text quantities. \\
\bottomrule
\end{tabular}
\end{center}
\endgroup

\section{Proofs and Additional Details}
\label{app:details-and-proofs}

\providecommand{\eps}{\varepsilon}

\subsection{Elementary structural facts}
\label{app:basic-structural-facts}

\subsubsection{One-step mean-field representations are not Bellman-sufficient}
\label{app:one-step-proof}

\begin{proof}[Proof of Proposition~\ref{prop:one-step-representation-not-sufficient}]
Take \(H=2\), \(\cX=\{x_0,u,v\}\), \(\cA=\{L,R\}\), and \(\mu_0=\delta_{x_0}\).
Let the stage-\(0\) reward and feature map be identically zero, so \(r_0^\star\equiv 0\) and \(f_0^\star\equiv 0\).
Let the stage-\(0\) transition be independent of the representation and satisfy \(T_0^\star(u\mid x,L,m)=1\) and \(T_0^\star(v\mid x,R,m)=1\) for every \(x\in\cX\) and \(m\in\RR\).

At stage \(1\), use the scalar feature map \(f_1^\star(u,a)=1\) and \(f_1^\star(v,a)=f_1^\star(x_0,a)=0\) for every \(a\in\cA\).
Let \(\sigma(m)\coloneq\min\{1,\max\{0,m\}\}\), and define the terminal reward by \(r_1^\star(u,a,m)=\sigma(m)\) and \(r_1^\star(v,a,m)=r_1^\star(x_0,a,m)=0\).
The stage-\(1\) transition kernel is irrelevant.

Let \(\eta\) choose \(L\) at \(x_0\) with probability one, and let
\(\tilde\eta\) choose \(R\) at \(x_0\) with probability one.
Since \(f_0^\star\equiv0\), we have \(m_0^{\star,\eta}=m_0^{\star,\tilde\eta}=0\).
At stage \(1\), however, the population is concentrated at \(u\) under
\(\eta\) and at \(v\) under \(\tilde\eta\). Hence \(m_1^{\star,\eta}=1\), whereas \(m_1^{\star,\tilde\eta}=0\).

Now consider action \(L\) for the representative agent at stage \(0\).
This action sends the representative agent to \(u\), irrespective of the
population policy. Therefore, for every representative policy \(\pi\),
\[
q_0^{\star,\pi}\left(x_0,L,m_{0:1}^{\star,\eta}\right)=1,
\qquad
q_0^{\star,\pi}\left(x_0,L,m_{0:1}^{\star,\tilde\eta}\right)=0.
\]
Thus the two action values differ although the current inputs
\((x_0,L,m_0)\) coincide.
\end{proof}

\subsubsection{Existence of Nash equilibria}
\label{app:existence-proof}

\begin{proof}[Proof of Proposition~\ref{prop:existence-nash}]
For each \(h\in\{0,\ldots,H-1\}\), define the compact convex set
\[
\mathcal K_h\coloneq
\operatorname{conv}\left\{f_h^\star(x,a):(x,a)\in\cX\times\cA\right\}
\subseteq \RR^d,
\qquad
\mathcal K\coloneq \prod_{h=0}^{H-1}\mathcal K_h .
\]
The set \(\mathcal K\) is nonempty, compact, and convex.
Its elements are candidate mean-field representation paths \(m=(m_0,\ldots,m_{H-1})\).

Fix \(m\in\mathcal K\).
Against this frozen representation path, the representative agent faces the finite-horizon MDP with reward \(r_h^\star(x,a,m_h)\) and transition kernel \(T_h^\star(\cdot\mid x,a,m_h)\).
Let \(\mathcal D\coloneq \prod_{h=0}^{H-1}\Delta(\cX\times\cA)\) be the compact convex set of state-action occupancy flows.
Define the feasible occupancy correspondence under the frozen path \(m\) by
\[
\begin{aligned}
\operatorname{Occ}(m)\coloneq\bigl\{\psi\in\mathcal D:\;&
\sum_{a\in\cA}\psi_0(x,a)=\mu_0(x),\quad \forall x\in\cX,\;
\sum_{a'\in\cA}\psi_{h+1}(x',a')=\\
&\sum_{x\in\cX}\sum_{a\in\cA}\psi_h(x,a)T_h^\star(x'\mid x,a,m_h),
\quad \forall h\le H-2,\ \forall x'\in\cX\bigr\}.
\end{aligned}
\]

For every \(m\in\mathcal K\), the set \(\operatorname{Occ}(m)\) is nonempty, compact, and convex.
Indeed, any Markov policy induces a feasible occupancy flow, compactness follows from closedness inside \(\mathcal D\), and convexity follows because the flow constraints are affine in \(\psi\) once \(m\) is fixed.

For \(\psi\in\operatorname{Occ}(m)\), define the frozen-path value
\[
J_m(\psi)\coloneq
\sum_{h=0}^{H-1}\sum_{x\in\cX}\sum_{a\in\cA}\psi_h(x,a)r_h^\star(x,a,m_h).
\]
The optimal occupancy correspondence is \(\operatorname{OptOcc}(m)\coloneq\argmax_{\psi\in\operatorname{Occ}(m)} J_m(\psi)\).
Because \(J_m\) is linear in \(\psi\) and \(\operatorname{Occ}(m)\) is nonempty, compact, and convex, the set \(\operatorname{OptOcc}(m)\) is nonempty, compact, and convex.

We next prove the closed-graph property of \(m\mapsto\operatorname{OptOcc}(m)\).
For \(m\in\mathcal K\), set \(v_H^m(x)\coloneq 0\) and define the Bellman functions recursively for \(h=H-1,\ldots,0\) by
\[
q_h^m(x,a)\coloneq r_h^\star(x,a,m_h)
+\sum_{x'\in\cX}T_h^\star(x'\mid x,a,m_h)v_{h+1}^m(x'),
\qquad
v_h^m(x)\coloneq\max_{a\in\cA}q_h^m(x,a).
\]
Since \(\cX\) and \(\cA\) are finite and the reward and transition maps are continuous in \(m\), backward induction shows that \(m\mapsto q_h^m(x,a)\) and \(m\mapsto v_h^m(x)\) are continuous for every \(h,x,a\).

For any \(\psi\in\operatorname{Occ}(m)\), the flow constraints give the telescoping identity
\[
J_m(\psi)=\sum_{x\in\cX}\mu_0(x)v_0^m(x)
-
\sum_{h=0}^{H-1}\sum_{x\in\cX}\sum_{a\in\cA}
\psi_h(x,a)\left(v_h^m(x)-q_h^m(x,a)\right).
\]
The residuals \(v_h^m(x)-q_h^m(x,a)\) are nonnegative.
Moreover, a policy that chooses only maximizers of \(q_h^m(x,\cdot)\) attains zero residual.
Hence \(\psi\in\operatorname{Occ}(m)\) is optimal if and only if
\[
\sum_{h=0}^{H-1}\sum_{x\in\cX}\sum_{a\in\cA}
\psi_h(x,a)\left(v_h^m(x)-q_h^m(x,a)\right)=0.
\]
Now suppose \(m^n\to m\), \(\psi^n\to\psi\), and \(\psi^n\in\operatorname{OptOcc}(m^n)\).
Passing to the limit in the flow constraints, using continuity of \(T_h^\star(\cdot\mid x,a,m_h)\), gives \(\psi\in\operatorname{Occ}(m)\).
Passing to the limit in the Bellman residual display, using continuity of \(q_h^m\) and \(v_h^m\), gives the zero-residual condition above.
Therefore \(\psi\in\operatorname{OptOcc}(m)\).
Thus, \(m\mapsto\operatorname{OptOcc}(m)\) has closed graph.

Define the linear representation map
\[
\Phi:\mathcal D\to \mathcal K,
\qquad
\Phi(\psi)\coloneq
\left(\sum_{x \in \cX}\sum_{a \in \cA}\psi_h(x,a)f_h^\star(x,a)\right)_{h=0}^{H-1}.
\]
For each \(m\in\mathcal K\), define the representation best-response correspondence \(\Gamma(m)\coloneq\{\Phi(\psi):\psi\in\operatorname{OptOcc}(m)\}\).
For every \(m\), the set \(\Gamma(m)\) is nonempty, compact, and convex, since it is the image of the nonempty compact convex set \(\operatorname{OptOcc}(m)\) under the linear map \(\Phi\).

The correspondence \(\Gamma\) has closed graph.
Indeed, let \(m^n\to m\), \(z^n\to z\), and \(z^n\in\Gamma(m^n)\).
Choose \(\psi^n\in\operatorname{OptOcc}(m^n)\) with \(z^n=\Phi(\psi^n)\).
Since \(\mathcal D\) is compact, some subsequence \(\psi^{n_k}\) converges to a \(\psi\in\mathcal D\).
By the closed-graph property of \(\operatorname{OptOcc}\), we have \(\psi\in\operatorname{OptOcc}(m)\).
Since \(\Phi\) is continuous, \(z=\lim_{k\to\infty}z^{n_k}=\lim_{k\to\infty}\Phi(\psi^{n_k})=\Phi(\psi)\).
Thus \(z\in\Gamma(m)\).
Since \(\Gamma:\mathcal K\rightrightarrows\mathcal K\), \(\mathcal K\) is compact and metric, and \(\Gamma\) has compact values, closed graph is equivalent to upper hemicontinuity.
Hence \(\Gamma\) is upper hemicontinuous with nonempty compact convex values.

By Kakutani's fixed-point theorem, there exists \(m^\star\in\mathcal K\) such that \(m^\star\in\Gamma(m^\star)\).
Therefore there exists \(\psi^\star\in\operatorname{OptOcc}(m^\star)\) satisfying
\[
m_h^\star=\sum_{x \in \cX}\sum_{a \in \cA}\psi_h^\star(x,a)f_h^\star(x,a),
\qquad h=0,\ldots,H-1.
\]

We now convert \(\psi^\star\) into a Markov policy.  Define
\[
\eta_h^\star(a\mid x)\coloneq
\begin{cases}
\displaystyle\frac{\psi_h^\star(x,a)}{\sum_{b\in\cA}\psi_h^\star(x,b)},
& \text{if } \sum_{b\in\cA}\psi_h^\star(x,b)>0, \\[12pt]
\displaystyle\frac{1}{|\cA|},
& \text{if } \sum_{b\in\cA}\psi_h^\star(x,b)=0 .
\end{cases}
\]
We claim that the population state-action flow induced by \(\eta^\star\) in the original mean-field model is exactly \(\psi^\star\).
Let \((\bar\mu_h,\bar\nu_h,m_h^{\star,\eta^\star})\) denote the population flow generated by \(\eta^\star\).
At \(h=0\), \(\bar\nu_0(x,a)=\mu_0(x)\eta_0^\star(a\mid x)=\psi_0^\star(x,a)\).
Suppose \(\bar\nu_h=\psi_h^\star\).  Then
\[
m_h^{\star,\eta^\star}
=\sum_{x\in\cX}\sum_{a\in\cA}\bar\nu_h(x,a)f_h^\star(x,a)
=\sum_{x\in\cX}\sum_{a\in\cA}\psi_h^\star(x,a)f_h^\star(x,a)
=m_h^\star.
\]
Using the flow constraint for \(\psi^\star\in\operatorname{Occ}(m^\star)\),
\[
\bar\mu_{h+1}(x')
=\sum_{x\in\cX}\sum_{a\in\cA}\bar\nu_h(x,a)
T_h^\star(x'\mid x,a,m_h^{\star,\eta^\star})
=
\sum_{a'\in\cA}\psi_{h+1}^\star(x',a').
\]
By the definition of \(\eta_{h+1}^\star\), this implies \(\bar\nu_{h+1}(x',a')=\bar\mu_{h+1}(x')\eta_{h+1}^\star(a'\mid x')=\psi_{h+1}^\star(x',a')\), with the convention above when \(\bar\mu_{h+1}(x')=0\).
Thus, by induction, \(\nu_h^{\star,\eta^\star}=\psi_h^\star\) and \(m_h^{\star,\eta^\star}=m_h^\star\) for every \(h\).

Finally, let \(\pi\in\Pi\) be arbitrary.
Against the population policy \(\eta^\star\), the representative agent faces precisely the frozen MDP with representation path \(m^\star\).
Let \(\psi^{\pi,m^\star}\in \operatorname{Occ}(m^\star)\) be the occupancy flow induced by \(\pi\) in this frozen MDP.
Since \(\psi^\star\in\operatorname{OptOcc}(m^\star)\),
\[
J(\pi;\eta^\star)=J_{m^\star}(\psi^{\pi,m^\star})
\le J_{m^\star}(\psi^\star)=J(\eta^\star;\eta^\star).
\]
Since this holds for every \(\pi\in\Pi\), we have \(\sup_{\pi\in\Pi}J(\pi;\eta^\star)-J(\eta^\star;\eta^\star)\le 0\).
The reverse inequality is trivial by taking \(\pi=\eta^\star\).
Hence \(\gap(\eta^\star)=0\), so \(\eta^\star\) is a Nash equilibrium.
\end{proof}

\subsubsection{Support and distance facts}

\begin{lemma}[Support stability]
\label{lem:app-support-stability}
Under Assumptions~\ref{ass:boundedness} and~\ref{ass:continuity}, fix \(\theta\in\Theta\), \(h\), \((x,a)\), and \(x'\in\cX\).
Then either
\[
T_h^\theta(x'\mid x,a,m)=0
\quad\text{for every }m\in\cB_f,
\]
or
\[
T_h^\theta(x'\mid x,a,m)\ge\tau
\quad\text{for every }m\in\cB_f.
\]
Consequently, \(\operatorname{supp}T_h^\theta(\cdot\mid x,a,m)\) does not depend on \(m\in\cB_f\).
\end{lemma}

\begin{proof}
For fixed \(x'\), the coordinate map \(g(m)=T_h^\theta(x'\mid x,a,m)\) is continuous because
\[
|g(m)-g(m')|
\le\tv\!\left(T_h^\theta(\cdot\mid x,a,m),T_h^\theta(\cdot\mid x,a,m')\right)
\le L_{T,h}\|m-m'\|_2.
\]
The ball \(\cB_f\) is connected, so \(g(\cB_f)\) is connected.
Since \(g(\cB_f)\subseteq\{0\}\cup[\tau,1]\), its image lies entirely in one of these two components.
Applying the argument to every \(x'\in\cX\) proves the support claim.
\end{proof}

\begin{lemma}[Elementary distance bounds]
\label{lem:app-distance-bounds}
For any probability distributions \(p,q\) on \(\cX\),
\[
\tv(p,q)^2\le 2\hel^2(p,q).
\]
If \(p\) and \(q\) have the same support and every positive probability in both distributions is at least \(\tau\), then
\[
\hel^2(p,q)\le\frac{\tv(p,q)^2}{2\tau},
\qquad
\kl(p\|q)\le\frac{4}{\tau}\tv(p,q)^2.
\]
\end{lemma}

\begin{proof}
Factoring \(|p(x)-q(x)|\) and applying Cauchy--Schwarz gives
\[
\begin{aligned}
2\tv(p,q)
&=\sum_x|\sqrt{p(x)}-\sqrt{q(x)}|
       (\sqrt{p(x)}+\sqrt{q(x)})\\
&\le \left(2\hel^2(p,q)\right)^{1/2}
\left(2+2\sum_x\sqrt{p(x)q(x)}\right)^{1/2}
\le 2\sqrt{2\hel^2(p,q)},
\end{aligned}
\]
which proves the first inequality.

For the second claim, sum over the common support.
Since \((\sqrt{p(x)}+\sqrt{q(x)})^2\ge4\tau\),
\[
\hel^2(p,q)=\frac12\sum_x
\frac{(p(x)-q(x))^2}{(\sqrt{p(x)}+\sqrt{q(x)})^2}
\le\frac{\|p-q\|_1^2}{8\tau}
=\frac{\tv(p,q)^2}{2\tau}.
\]
Finally, \(\kl(p\|q)\le\chi^2(p\|q)\), and hence
\[
\kl(p\|q)\le\sum_x\frac{(p(x)-q(x))^2}{q(x)}
\le\frac{\|p-q\|_1^2}{\tau}
=\frac{4}{\tau}\tv(p,q)^2.
\]
\end{proof}

\subsection{Exact statistical confidence event}
\label{app:finite-confidence}

This subsection proves the precise finite-class version of Theorem~\ref{thm:finite-class-confidence}.
The empirical criterion uses negative log likelihood and squared reward loss.
The argument first controls the corresponding prediction error at the sampled representations and then incorporates the finite-population representation error.

\subsubsection{Episode quantities}

For episode \(n\) and stage \(h\), define \(P_{n,h}\coloneq T_h^\star(\cdot\mid Y_{n,h},m_h^{\star,\rho_n})\) and \(Q_{n,h}^\theta\coloneq T_h^\theta(\cdot\mid Y_{n,h},\widehat m_{n,h}^\theta)\), and let \(e_{n,h}^\theta\coloneq R_{n,h}-r_h^\theta(Y_{n,h},\widehat m_{n,h}^\theta)\).
The one-episode likelihood difference and prediction loss are
\[
A_n(\theta)\coloneq
\sum_{h=0}^{H-1}\left[
\log\frac{P_{n,h}(X_{n,h+1})}{Q_{n,h}^\theta(X_{n,h+1})}
+(e_{n,h}^\theta)^2\right],
\]
and
\[
W_n(\theta)\coloneq
\sum_{h=0}^{H-1}\left[
2\hel^2(P_{n,h},Q_{n,h}^\theta)+(e_{n,h}^\theta)^2\right].
\]
The log ratio is \(+\infty\) when its numerator is positive and its denominator is zero.
Its value when both probabilities are zero is immaterial, since \(X_{n,h+1}\) is drawn from \(P_{n,h}\).
Write \(\overline A_N(\theta)=N^{-1}\sum_{n=1}^N A_n(\theta)\) and \(\overline W_N(\theta)=N^{-1}\sum_{n=1}^N W_n(\theta)\), and define the appendix-only quantity \(\tilde\cL(\theta)\coloneq\EE W_1(\theta)\).

The empirical criterion decomposes as
\[
\widehat\cR(\theta)=\widehat H_N+\overline A_N(\theta),
\qquad
\widehat H_N\coloneq\frac1N\sum_{n=1}^N\sum_{h=0}^{H-1}
-\log P_{n,h}(X_{n,h+1}).
\]
Every observed next state has positive probability under \(P_{n,h}\), and that probability is at least \(\tau\).
Thus \(\widehat H_N<\infty\) almost surely, and the decomposition is valid in the extended reals.
The quantity \(\widehat H_N\) does not depend on \(\theta\).

\subsubsection{Likelihood and prediction loss}

\begin{lemma}[Likelihood-to-Hellinger bound]
\label{lem:app-likelihood-hellinger}
Let \((\mathcal F_i)\) be a filtration, and suppose that \(Z_i\), conditionally on \(\mathcal F_{i-1}\), has probability mass function \(p_i\).
For each \(\theta\in\Theta\), let \(q_i^\theta\) be an \(\mathcal F_{i-1}\)-measurable probability mass function.
Then, for every \(\delta\in(0,1)\), with probability at least \(1-\delta\), simultaneously for every \(\theta\in\Theta\),
\[
\sum_i\log\frac{p_i(Z_i)}{q_i^\theta(Z_i)}
\ge 2\sum_i\hel^2(p_i,q_i^\theta)-2\log\frac{|\Theta|}{\delta}.
\]
The log ratio is interpreted \(p_i\)-almost surely, with value \(+\infty\) when \(p_i(Z_i)>0\) and \(q_i^\theta(Z_i)=0\).
\end{lemma}

\begin{proof}
Fix \(\theta\). With the convention \(e^{-\infty}=0\),
\[
\begin{aligned}
\EE\left[\exp\left(-\frac12\log\frac{p_i(Z_i)}{q_i^\theta(Z_i)}\right)
\middle|\mathcal F_{i-1}\right]
&=\sum_z\sqrt{p_i(z)q_i^\theta(z)}\\
&=1-\hel^2(p_i,q_i^\theta)\le e^{-\hel^2(p_i,q_i^\theta)}.
\end{aligned}
\]
Consequently,
\[
\exp\left(
\sum_{j\le i}\left[
\hel^2(p_j,q_j^\theta)-\frac12\log\frac{p_j(Z_j)}{q_j^\theta(Z_j)}\right]
\right)
\]
is a nonnegative supermartingale.
Markov's inequality followed by a union bound over \(\Theta\) proves the claim.
\end{proof}

Fix \(\delta\in(0,1)\) and set \(u\coloneq\log(3|\Theta|/\delta)\).
For each episode, reveal \(\rho_n\) and the entire population-snapshot collection before revealing the representative trajectory chronologically.
The conditional independence in the offline data model ensures that this does not change the representative trajectory law.
Immediately before \(X_{n,h+1}\) is revealed, \(Y_{n,h}\), \(P_{n,h}\), and \(Q_{n,h}^\theta\) are measurable, and the conditional law of \(X_{n,h+1}\) is \(P_{n,h}\).
Applying Lemma~\ref{lem:app-likelihood-hellinger} across all episode-stage pairs and adding the reward errors gives, with probability at least \(1-\delta/3\),
\begin{equation}
\label{eq:app-likelihood-event}
\overline A_N(\theta)\ge\overline W_N(\theta)-\frac{2u}{N}
\qquad\text{for every }\theta\in\Theta.
\end{equation}

\begin{lemma}[Empirical prediction loss]
\label{lem:app-empirical-prediction-loss}
With probability at least \(1-\delta/3\),
\[
\tilde\cL(\theta)\le2\overline W_N(\theta)+\frac{8Hu}{N}
\qquad\text{for every }\theta\in\Theta.
\]
\end{lemma}

\begin{proof}
For every \(\theta\), we have \(0\le W_n(\theta)\le3H\) and \(W_n(\theta)^2\le3H\,W_n(\theta)\).
Let \(\mu_\theta=\EE W_1(\theta)\).
A lower-tail Bernstein inequality gives, with probability at least \(1-e^{-u}\),
\[
\mu_\theta-\overline W_N(\theta)
\le\sqrt{\frac{6H\mu_\theta u}{N}}+\frac{Hu}{N}
\le\frac{\mu_\theta}{2}+\frac{4Hu}{N}.
\]
Rearranging and taking a union bound over \(\Theta\) proves the result.
\end{proof}

\subsubsection{The realizable model}

Set \(M\coloneq\log(1/\tau)\) and
\[
b_K\coloneq\frac{B_f^2}{K}\sum_{h=0}^{H-1}
\left(\frac{L_{T,h}^2}{\tau}+L_{r,h}^2\right).
\]
The exact representation \(m_h^{\star,\rho}\), every \(\widehat m_{n,h}^\theta\), and every \(\bar m_h^{\theta,\rho}\) lie in \(\cB_f\), since they are averages of vectors with norm at most \(B_f\).
For \(\theta^\star\), realizability and Lemma~\ref{lem:app-support-stability} therefore imply that \(P_{n,h}\) and \(Q_{n,h}^{\theta^\star}\) have the same support.
On this support, all positive probabilities are at least \(\tau\), so
\begin{equation}
\label{eq:app-true-log-bound}
\left|\log\frac{P_{n,h}(x')}{Q_{n,h}^{\theta^\star}(x')}\right|\le M.
\end{equation}
In particular, the true model has finite empirical risk almost surely.

\begin{lemma}[Realizable score]
\label{lem:app-realizable-score}
There is a numerical constant \(C_0\) such that, with probability at least \(1-\delta/3\),
\[
\overline A_N(\theta^\star)\le
C_0\left(b_K+\frac{H(1+M)u}{N}\right).
\]
\end{lemma}

\begin{proof}
Conditioning on the logged state-action pair and the population snapshot gives
\[
\begin{aligned}
\EE A_1(\theta^\star)
&=\sum_{h=0}^{H-1}\EE\left[
\kl(P_{1,h}\|Q_{1,h}^{\theta^\star})+(e_{1,h}^{\theta^\star})^2\right].
\end{aligned}
\]
For a generic population snapshot, conditional on \(\rho\),
\[
\begin{aligned}
\EE\left[
\|\widehat m_h^{\theta,\rho}-\bar m_h^{\theta,\rho}\|_2^2
\middle|\rho\right]
&=\frac1K\operatorname{tr}\operatorname{Cov}_{\nu_h^{\star,\rho}}
(f_h^\theta(X,A))\le\frac{B_f^2}{K}.
\end{aligned}
\]
For the realizable model, \(\bar m_h^{\theta^\star,\rho}=m_h^{\star,\rho}\).
By Lemmas~\ref{lem:app-support-stability} and~\ref{lem:app-distance-bounds}, the total-variation Lipschitz condition, and reward Lipschitzness,
\[
\EE A_1(\theta^\star)
\le\frac{B_f^2}{K}\sum_{h=0}^{H-1}
\left(\frac{4L_{T,h}^2}{\tau}+L_{r,h}^2\right)
\le4b_K.
\]

We next record a log-moment bound.
If \(p,q\) have common support and \(|\log(p/q)|\le M\), then
\begin{equation}
\label{eq:app-log-second-moment}
\EE_{X\sim p}\left[\left(\log\frac{p(X)}{q(X)}\right)^2\right]
\le16(1+M)\kl(p\|q).
\end{equation}
Indeed, write \(x=p(y)/q(y)\). For \(0<x\le e^M\),
we have \(x(\log x)^2/(x\log x-x+1)\le16(1+M)\),
where the ratio is interpreted continuously at \(x=1\).
To see this, put \(x=e^t\).
The ratio is at most \(2\) when \(t\le0\), at most \(2e^2\) when \(0\le t\le2\), and at most \(2t\) when \(t\ge2\).
Multiplying by \(q(y)(x\log x-x+1)\), summing over \(y\), and using \(\kl(p\|q)=\sum_yq(y)\bigl(x\log x-x+1\bigr)\) proves \eqref{eq:app-log-second-moment}.

Using \eqref{eq:app-true-log-bound}, \eqref{eq:app-log-second-moment}, \(0\le(e_{1,h}^{\theta^\star})^2\le1\), and \((\sum_h z_h)^2\le H\sum_hz_h^2\), we obtain
\(\EE A_1(\theta^\star)^2\le128H(1+M)b_K\).
Also, \(|A_1(\theta^\star)|\le H(1+M)\).
Bernstein's inequality across the \(N\) independent episodes, followed by Young's inequality, now gives
\[
\overline A_N(\theta^\star)\le
C_0\left(b_K+\frac{H(1+M)u}{N}\right)
\]
with probability at least \(1-e^{-u}\), which is at least \(1-\delta/3\).
\end{proof}

\subsubsection{Finite-population error}

\begin{lemma}[Finite-population error]
\label{lem:app-plugin-to-exact-loss}
For every \(\theta\in\Theta\),
\[
\cL(\theta)\le2\tilde\cL(\theta)+2b_K.
\]
\end{lemma}

\begin{proof}
Fix \(h,\rho,x,a\), and let \(\widehat m_h^{\theta,\rho}\) denote a generic size-\(K\) snapshot average.
Write \(P=T_h^\star(\cdot\mid x,a,m_h^{\star,\rho})\), \(Q=T_h^\theta(\cdot\mid x,a,\widehat m_h^{\theta,\rho})\), and \(\bar Q=T_h^\theta(\cdot\mid x,a,\bar m_h^{\theta,\rho})\).
Since Hellinger distance is a metric, \(2\hel^2(P,\bar Q)\le4\hel^2(P,Q)+4\hel^2(Q,\bar Q)\).
The two candidate representations lie in \(\cB_f\), so Lemma~\ref{lem:app-support-stability} gives \(Q\) and \(\bar Q\) the same support.
Lemma~\ref{lem:app-distance-bounds} and Assumption~\ref{ass:continuity} then give
\[
\hel^2(Q,\bar Q)\le\frac{\tv(Q,\bar Q)^2}{2\tau}
\le\frac{L_{T,h}^2}{2\tau}
\|\widehat m_h^{\theta,\rho}-\bar m_h^{\theta,\rho}\|_2^2.
\]
For rewards,
\[
\begin{aligned}
&\left(r_h^\star(x,a,m_h^{\star,\rho})
-r_h^\theta(x,a,\bar m_h^{\theta,\rho})\right)^2\\
&\quad\le2\left(r_h^\star(x,a,m_h^{\star,\rho})
-r_h^\theta(x,a,\widehat m_h^{\theta,\rho})\right)^2+2L_{r,h}^2
\|\widehat m_h^{\theta,\rho}-\bar m_h^{\theta,\rho}\|_2^2.
\end{aligned}
\]
Averaging over the snapshot, the logged state-action pair, and \(\rho\), then summing over \(h\), proves the claim using \(\EE[\|\widehat m_h^{\theta,\rho}-\bar m_h^{\theta,\rho}\|_2^2\mid\rho]\le B_f^2/K\).
\end{proof}

\subsubsection{Confidence theorem}

\begin{theorem}[Finite-class confidence event]
\label{thm:app-finite-class-confidence}
Assume \(|\Theta|<\infty\) and Assumptions~\ref{ass:realizability}--\ref{ass:continuity}.
Fix \(\delta\in(0,1)\), and let \(u,M,b_K\) be as defined above.
For numerical constants \(C_\beta,C_s\), set
\[
\beta\coloneq C_\beta\left(b_K+\frac{H(1+M)u}{N}\right),
\qquad
s_{N,K}\coloneq C_s\left(b_K+\frac{H(1+M)u}{N}\right).
\]
For sufficiently large \(C_\beta,C_s\), with probability at least \(1-\delta\),
\[
\theta^\star\in\cC(\beta),
\qquad
\cC(\beta)\subseteq\Theta_{s_{N,K}}.
\]
In particular,
\[
\beta,\ s_{N,K}=O\!\left(
\frac{H\bigl(1+\log(1/\tau)\bigr)\log(|\Theta|/\delta)}{N}
+
\frac{B_f^2}{K}\sum_{h=0}^{H-1}
\left[\frac{L_{T,h}^2}{\tau}+L_{r,h}^2\right]\right).
\]
\end{theorem}

\begin{proof}
Intersect the events in \eqref{eq:app-likelihood-event} and Lemmas~\ref{lem:app-empirical-prediction-loss} and~\ref{lem:app-realizable-score}.
Their intersection has probability at least \(1-\delta\).

First, \eqref{eq:app-likelihood-event} and \(\overline W_N(\theta)\ge0\) imply \(\inf_{\theta\in\Theta}\overline A_N(\theta)\ge-2u/N\).
Let \(r_{N,K}\coloneq b_K+H(1+M)u/N\).
The decomposition of \(\widehat\cR\) and Lemma~\ref{lem:app-realizable-score} give
\[
\begin{aligned}
\widehat\cR(\theta^\star)-\inf_{\theta\in\Theta}\widehat\cR(\theta)
&=\overline A_N(\theta^\star)-\inf_{\theta\in\Theta}\overline A_N(\theta)\\
&\le C_0r_{N,K}+\frac{2u}{N}
\le(C_0+2)r_{N,K}.
\end{aligned}
\]
Thus \(\theta^\star\in\cC(\beta)\) when \(C_\beta\ge C_0+2\).

Now fix \(\theta\in\cC(\beta)\).
Since \(\inf_{\theta'}\widehat\cR(\theta')\le \widehat\cR(\theta^\star)\),
we have \(\overline A_N(\theta)\le\overline A_N(\theta^\star)+\beta\).
Using \eqref{eq:app-likelihood-event} once more, \(\overline W_N(\theta)\le\overline A_N(\theta)+2u/N\le(C_0+C_\beta+2)r_{N,K}\).
Lemma~\ref{lem:app-empirical-prediction-loss} then yields \(\tilde\cL(\theta)\le2(C_0+C_\beta+2)r_{N,K}+8Hu/N\le2(C_0+C_\beta+6)r_{N,K}\).
Finally, Lemma~\ref{lem:app-plugin-to-exact-loss} gives \(\cL(\theta)\le[4(C_0+C_\beta+6)+2]r_{N,K}\).
Choosing \(C_s\) at least this large proves \(\theta\in\Theta_{s_{N,K}}\).
The order bound follows from the definitions of \(u,M,b_K\).
\end{proof}

\subsection{Interval gaps and deterministic oracle inequalities}
\label{app:interval-gap}

For an arbitrary set \(\cC\subseteq\Theta\), define
\[
\overline J_{\cC}(\pi;\eta)\coloneq\sup_{\theta\in\cC}J_\theta(\pi;\eta),
\qquad
\underline J_{\cC}(\pi;\eta)\coloneq\inf_{\theta\in\cC}J_\theta(\pi;\eta),
\]
and
\[
\widehat\gap_{\cC}(\eta)\coloneq
\sup_{\pi\in\Pi}\overline J_{\cC}(\pi;\eta)-\underline J_{\cC}(\eta;\eta).
\]
When \(\cC=\cC(\beta)\), these agree with the main-text quantities \(\overline J_\beta\), \(\underline J_\beta\), and \(\widehat\gap_\beta\).
For \(\theta^\star\in\cC\), define the one-sided interval errors
\[
r^+_{\cC}(\pi;\eta)\coloneq\overline J_{\cC}(\pi;\eta)-J(\pi;\eta),
\qquad
r^-_{\cC}(\pi;\eta)\coloneq J(\pi;\eta)-\underline J_{\cC}(\pi;\eta).
\]
They are nonnegative.  Define also
\[
\operatorname{subopt}^{+}_{\cC,\eta}(\tilde\pi)\coloneq
\sup_{\pi\in\Pi}\overline J_{\cC}(\pi;\eta)-\overline J_{\cC}(\tilde\pi;\eta).
\]
When \(\cC=\cC(\beta)\), this is \(\operatorname{subopt}^{+}_{\beta,\eta}\).

\begin{lemma}[Interval objective upper-bounds the true gap]
\label{lem:app-interval-upper-gap}
If \(\theta^\star\in\cC\), then for every \(\eta\in\Pi\),
\[
\gap(\eta)\le\widehat\gap_{\cC}(\eta).
\]
\end{lemma}

\begin{proof}
The claim follows from \(J(\pi;\eta)=J_{\theta^\star}(\pi;\eta)\le\overline J_{\cC}(\pi;\eta)\) and \(J(\eta;\eta)=J_{\theta^\star}(\eta;\eta)\ge\underline J_{\cC}(\eta;\eta)\).
Taking the supremum over \(\pi\) gives \(\gap(\eta)\le\widehat\gap_{\cC}(\eta)\).
\end{proof}

\begin{theorem}[Deterministic interval oracle inequality]
\label{thm:app-deterministic-interval}
Let \(\cC\subseteq\Theta\) satisfy \(\theta^\star\in\cC\).
Let \(\widehat\eta\in\Pi\) obey
\[
\widehat\gap_{\cC}(\widehat\eta)
\le\inf_{\eta\in\Pi}\widehat\gap_{\cC}(\eta)+\eopt.
\]
Then, for every comparator \(\eta\in\Pi\),
\[
\gap(\widehat\eta)\le\eopt+\gap(\eta)+r^-_{\cC}(\eta;\eta)
+\inf_{\tilde\pi\in\Pi}\left\{
r^+_{\cC}(\tilde\pi;\eta)+\operatorname{subopt}^{+}_{\cC,\eta}(\tilde\pi)
\right\}.
\]
\end{theorem}

\begin{proof}
By Lemma~\ref{lem:app-interval-upper-gap} and approximate optimality,
\[
\gap(\widehat\eta)\le\widehat\gap_{\cC}(\widehat\eta)
\le\widehat\gap_{\cC}(\eta)+\eopt.
\]
Fix \(\tilde\pi\in\Pi\).
By definition of \(\operatorname{subopt}^{+}_{\cC,\eta}\),
\[
\widehat\gap_{\cC}(\eta)=\overline J_{\cC}(\tilde\pi;\eta)
-\underline J_{\cC}(\eta;\eta)+\operatorname{subopt}^{+}_{\cC,\eta}(\tilde\pi).
\]
Using the one-sided errors, \(\overline J_{\cC}(\tilde\pi;\eta)\le J(\tilde\pi;\eta)+r^+_{\cC}(\tilde\pi;\eta)\) and \(-\underline J_{\cC}(\eta;\eta)\le-J(\eta;\eta)+r^-_{\cC}(\eta;\eta)\).
Thus
\[
\widehat\gap_{\cC}(\eta)\le J(\tilde\pi;\eta)-J(\eta;\eta)
+r^+_{\cC}(\tilde\pi;\eta)+r^-_{\cC}(\eta;\eta)+
\operatorname{subopt}^{+}_{\cC,\eta}(\tilde\pi).
\]
Since \(J(\tilde\pi;\eta)-J(\eta;\eta)\le\gap(\eta)\), taking the infimum over \(\tilde\pi\) completes the proof.
\end{proof}


\subsection{Bellman residuals and value-transfer radii}
\label{app:bellman-value-transfer}

For \(\theta\in\Theta\) and \(\pi\in\Pi\), define \(v_H^{\theta,\pi}(x,\varnothing)=0\) and, recursively,
\[
q_h^{\theta,\pi}(x,a,z_{h:H-1})=r_h^\theta(x,a,z_h)
+
\sum_{x'}T_h^\theta(x'\mid x,a,z_h)v_{h+1}^{\theta,\pi}(x',z_{h+1:H-1}),
\]
and let \(v_h^{\theta,\pi}(x,z_{h:H-1})=\sum_a\pi_h(a\mid x)q_h^{\theta,\pi}(x,a,z_{h:H-1})\).
For \(m\in\cB_f\) and a candidate suffix \(z_{h:H-1}\), define
\[
\begin{aligned}
\mathfrak b_h^{\theta,\pi}(x,a;m,z_{h:H-1})
={}&q_h^{\theta,\pi}(x,a,z_{h:H-1})-r_h^\star(x,a,m)\\
&-\sum_{x'}T_h^\star(x'\mid x,a,m)v_{h+1}^{\theta,\pi}(x',z_{h+1:H-1}).
\end{aligned}
\]
For a finite positive measure \(\mathsf M\) on such tuples, write
\[
\|\mathfrak b^{\theta,\pi}\|_{2,\mathsf M}^2
=\int \mathfrak b_h^{\theta,\pi}(x,a;m,z_{h:H-1})^2
\,d\mathsf M(h,m,z_{h:H-1},x,a).
\]
Define
\[
\mathsf D_{\mathrm{log}}^\theta=\sum_{h=0}^{H-1}\EE_{\rho\sim\fp}\left[
\delta_{(h,m_h^{\star,\rho},\bar m_{h:H-1}^{\theta,\rho})}
\otimes \nu_h^{\star,\rho}\right],
\]
\[
\mathsf D_{\mathrm{tar}}^{\theta,\pi;\eta}=\sum_{h=0}^{H-1}
\delta_{(h,m_h^{\star,\eta},m_{h:H-1}^{\theta,\eta})}
\otimes \nu_h^{\star,\pi;\eta}.
\]
Put \(\gamma_h=H-h-1\) and \(a_H=1+(H-1)^2\).

\begin{lemma}[Bellman-residual value identity]
\label{lem:app-bellman-identity}
For every \(\theta\in\Theta\) and every \(\pi,\eta\in\Pi\),
\[
J_\theta(\pi;\eta)-J(\pi;\eta)
=\int \mathfrak b^{\theta,\pi}\,d\mathsf D_{\mathrm{tar}}^{\theta,\pi;\eta}.
\]
\end{lemma}

\begin{proof}
Fix \(\theta,\pi,\eta\), and write \(z_{h:H-1}=m_{h:H-1}^{\theta,\eta}\). The candidate value can be written as
\[
J_\theta(\pi;\eta)=\EE_{X_0\sim\mu_0}v_0^{\theta,\pi}(X_0,z_{0:H-1}).
\]
Under the true law \(\PP_\star^{\pi;\eta}\), the population path is \(m^{\star,\eta}\). Since \(v_H^{\theta,\pi}\equiv0\), telescoping gives
\[
\begin{aligned}
J_\theta(\pi;\eta)-J(\pi;\eta)
=\sum_{h=0}^{H-1}\EE_\star^{\pi;\eta}\Big[&
v_h^{\theta,\pi}(X_h,z_{h:H-1})-r_h^\star(X_h,A_h,m_h^{\star,\eta})\\
&-v_{h+1}^{\theta,\pi}(X_{h+1},z_{h+1:H-1})\Big].
\end{aligned}
\]
Because \(A_h\sim\pi_h(\cdot\mid X_h)\),
\[
\EE_{\star}^{\pi;\eta}\!\left[v_h^{\theta,\pi}(X_h,z_{h:H-1})\right]
=\EE_{\star}^{\pi;\eta}\!\left[q_h^{\theta,\pi}(X_h,A_h,z_{h:H-1})\right].
\]
Conditioning the last term on \((X_h,A_h)\) and using \(X_{h+1}\sim T_h^\star(\cdot\mid X_h,A_h,m_h^{\star,\eta})\) gives exactly
\[
\sum_{h=0}^{H-1}\EE_{(X,A)\sim\nu_h^{\star,\pi;\eta}}
\left[\mathfrak b_h^{\theta,\pi}
(X,A;m_h^{\star,\eta},m_{h:H-1}^{\theta,\eta})\right],
\]
which is the displayed integral.
\end{proof}

\begin{lemma}[Logged Bellman residual is controlled by prediction loss]
\label{lem:app-logged-bellman-loss}
For every \(\theta\in\Theta\) and every \(\pi\in\Pi\),
\[
\|\mathfrak b^{\theta,\pi}\|_{2,\mathsf D_{\mathrm{log}}^\theta}^2
\le a_H\cL(\theta).
\]
\end{lemma}

\begin{proof}
Fix \(h,\rho,x,a\), and write \(z_{h:H-1}=\bar m_{h:H-1}^{\theta,\rho}\), \(m=m_h^{\star,\rho}\). Since rewards lie in \([0,1]\), backward induction gives
\(0\le v_{h+1}^{\theta,\pi}(x',z_{h+1:H-1})\le\gamma_h\) for all \(x'\).
Thus, with \(P=T_h^\star(\cdot\mid x,a,m)\) and \(Q=T_h^\theta(\cdot\mid x,a,z_h)\),
\[
\left|\sum_{x'}(Q(x')-P(x'))v_{h+1}^{\theta,\pi}(x',z_{h+1:H-1})\right|
\le\gamma_h\tv(P,Q).
\]
Therefore,
\[
\left|\mathfrak b_h^{\theta,\pi}(x,a;m,z_{h:H-1})\right|
\le\left|r_h^\theta(x,a,z_h)-r_h^\star(x,a,m)\right|+\gamma_h\tv(P,Q).
\]
By Cauchy--Schwarz in \(\RR^2\) and Lemma~\ref{lem:app-distance-bounds},
\[
\begin{aligned}
\mathfrak b_h^{\theta,\pi}(x,a;m,z_{h:H-1})^2
&\le(1+\gamma_h^2)\left[
\left(r_h^\theta(x,a,z_h)-r_h^\star(x,a,m)\right)^2
+\tv(P,Q)^2\right]\\
&\le(1+\gamma_h^2)\left[
\left(r_h^\theta(x,a,z_h)-r_h^\star(x,a,m)\right)^2
+2\hel^2(P,Q)\right].
\end{aligned}
\]
Since \(1+\gamma_h^2\le a_H\), integrating this display over \(\rho\sim\fp\), \((x,a)\sim\nu_h^{\star,\rho}\), and summing over \(h\) gives the result by the definition of \(\cL(\theta)\).
\end{proof}

For \(s\ge0\), define the exact Bellman-transfer radius
\[
\operatorname{rad}_{\mathrm{Bell},s}(\pi;\eta)\coloneq
\sup_{\theta\in\Theta_s}\left|
\int \mathfrak b^{\theta,\pi}\,d\mathsf D_{\mathrm{tar}}^{\theta,\pi;\eta}
\right|,
\]
with the convention that the supremum of an empty set is zero. The coefficient used in the main text is
\[
\kappa(\pi;\eta;s)=\sup_{\theta\in\Theta_s}\frac{
\left(\displaystyle\int \mathfrak b^{\theta,\pi}\,d\mathsf D_{\mathrm{tar}}^{\theta,\pi;\eta}\right)^2
}{\|\mathfrak b^{\theta,\pi}\|_{2,\mathsf D_{\mathrm{log}}^\theta}^2},
\]
with \(0/0=0\) and positive-over-zero equal to \(+\infty\). For \(s>0\), set \(R_s(\pi;\eta)=\sqrt{a_Hs\,\kappa(\pi;\eta;s)}\). Then \(\operatorname{rad}_{\mathrm{Bell},s}(\pi;\eta)\le R_s(\pi;\eta)\).
If \(\kappa(\pi;\eta;s)=+\infty\), the claim is immediate.
Otherwise, the zero-denominator conventions imply that a zero logged residual must also have zero target integral.
Hence, for every \(\theta\in\Theta_s\), the definition of \(\kappa\) and Lemma~\ref{lem:app-logged-bellman-loss} give
\[
\left(\int \mathfrak b^{\theta,\pi}\,d\mathsf D_{\mathrm{tar}}^{\theta,\pi;\eta}\right)^2
\le\kappa(\pi;\eta;s)\|\mathfrak b^{\theta,\pi}\|_{2,\mathsf D_{\mathrm{log}}^\theta}^2
\le a_Hs\,\kappa(\pi;\eta;s).
\]

\begin{lemma}[One-sided interval control by Bellman-transfer radii]
\label{lem:app-one-sided-radii}
Suppose \(\theta^\star\in\cC\subseteq\Theta_s\). Then, for every \(\pi,\eta\in\Pi\),
\[
r^+_{\cC}(\pi;\eta)\le\operatorname{rad}_{\mathrm{Bell},s}(\pi;\eta),
\qquad
r^-_{\cC}(\pi;\eta)\le\operatorname{rad}_{\mathrm{Bell},s}(\pi;\eta).
\]
Consequently, for \(s>0\),
\[
r^+_{\cC}(\pi;\eta)\le R_s(\pi;\eta),
\qquad
r^-_{\cC}(\pi;\eta)\le R_s(\pi;\eta).
\]
\end{lemma}

\begin{proof}
For the upper side,
\[
\begin{aligned}
r^+_{\cC}(\pi;\eta)
&=\sup_{\theta\in\cC}\{J_\theta(\pi;\eta)-J(\pi;\eta)\} \\
&\le\sup_{\theta\in\cC}|J_\theta(\pi;\eta)-J(\pi;\eta)| \\
&=\sup_{\theta\in\cC}\left|
\int \mathfrak b^{\theta,\pi}\,d\mathsf D_{\mathrm{tar}}^{\theta,\pi;\eta}
\right|\le\operatorname{rad}_{\mathrm{Bell},s}(\pi;\eta),
\end{aligned}
\]
where the equality is Lemma~\ref{lem:app-bellman-identity}. The lower side is identical. The coefficient form follows from \(\operatorname{rad}_{\mathrm{Bell},s}\le R_s\).
\end{proof}

\subsection{Proof of the adaptive Nash-gap theorem and rate corollary}
\label{app:nash-gap-proof}

\begin{proof}[Proof of Theorem~\ref{thm:adaptive-nash-gap}]
Apply Theorem~\ref{thm:app-deterministic-interval} with \(\cC=\cC(\beta)\).
Since \(\theta^\star\in\cC(\beta)\subseteq\Theta_s\), Lemma~\ref{lem:app-one-sided-radii} gives \(r^-_{\cC(\beta)}(\eta;\eta)\le R_s(\eta;\eta)\).
For every \(\tilde\pi\in\Pi\), it also gives \(r^+_{\cC(\beta)}(\tilde\pi;\eta)\le R_s(\tilde\pi;\eta)\).
Substituting these two inequalities into the deterministic oracle inequality proves
\[
\gap(\widehat\eta)\le\eopt+\gap(\eta)+R_s(\eta;\eta)
+\inf_{\tilde\pi\in\Pi}
\left\{R_s(\tilde\pi;\eta)+
\operatorname{subopt}^{+}_{\beta,\eta}(\tilde\pi)\right\}.
\]
\end{proof}

\begin{proof}[Proof of Corollary~\ref{cor:equilibrium-rate}]
On the event of Theorem~\ref{thm:app-finite-class-confidence}, \(\theta^\star\in\cC(\beta)\subseteq\Theta_{s_{N,K}}\).
Theorem~\ref{thm:adaptive-nash-gap} therefore gives, for every \(\eta\in\Pi\),
\[
\gap(\widehat\eta)\le\eopt+\gap(\eta)+R_{s_{N,K}}(\eta;\eta)
+\inf_{\tilde\pi\in\Pi}\left\{R_{s_{N,K}}(\tilde\pi;\eta)+
\operatorname{subopt}^{+}_{\beta,\eta}(\tilde\pi)\right\}.
\]
For fixed \(\eta\), choose a sequence \((\tilde\pi_j)_{j\ge1}\) such that \(\operatorname{subopt}^{+}_{\beta,\eta}(\tilde\pi_j)\downarrow0\); such a sequence exists by the definition of the supremum in \(\operatorname{subopt}^{+}_{\beta,\eta}\).
Since \(R_{s_{N,K}}(\tilde\pi_j;\eta)\le\sqrt{a_Hs_{N,K}}\sup_{\pi\in\Pi}\sqrt{\kappa(\pi;\eta;s_{N,K})}\), letting \(j\to\infty\) gives
\[
\gap(\widehat\eta)\le\eopt+\gap(\eta)+
\sqrt{a_Hs_{N,K}}\sqrt{\kappa(\eta;\eta;s_{N,K})}
+\sqrt{a_Hs_{N,K}}\sup_{\pi\in\Pi}\sqrt{\kappa(\pi;\eta;s_{N,K})}.
\]
Taking the infimum over \(\eta\in\Pi\) proves the first display of the corollary. If \(\eta_\star\in\Pi_\star\), then \(\gap(\eta_\star)=0\), and substituting \(\eta=\eta_\star\) yields
\[
\gap(\widehat\eta)\le\eopt+2\sqrt{a_Hs_{N,K}}\max\left\{
\sqrt{\kappa(\eta_\star;\eta_\star;s_{N,K})},
\sup_{\pi\in\Pi}\sqrt{\kappa(\pi;\eta_\star;s_{N,K})}
\right\}.
\]
\end{proof}

\subsection{Scalar and \texorpdfstring{\(L^2\)}{L2} Bellman coverage}
\label{app:kappa-coverage}

Define
\[
\mathscr{C}(\pi;\eta;s)=\sup_{\theta\in\Theta_s}\frac{
\|\mathfrak b^{\theta,\pi}\|_{2,\mathsf D_{\mathrm{tar}}^{\theta,\pi;\eta}}^2
}{\|\mathfrak b^{\theta,\pi}\|_{2,\mathsf D_{\mathrm{log}}^\theta}^2},
\]
with \(0/0=0\) and positive-over-zero equal to \(+\infty\).

\begin{proposition}[\(L^2\) Bellman coverage controls scalar Bellman transfer]
\label{prop:app-kappa-vs-l2-coverage}
For every \(\pi,\eta\in\Pi\) and every \(s\ge0\),
\[
\kappa(\pi;\eta;s)\le H\mathscr{C}(\pi;\eta;s).
\]
\end{proposition}

\begin{proof}
Fix \(\theta\in\Theta_s\). Since \(\mathsf D_{\mathrm{tar}}^{\theta,\pi;\eta}\) is a sum of \(H\) probability measures, its total mass is \(H\). Cauchy-Schwarz inequality gives
\[
\left(\int \mathfrak b^{\theta,\pi}\,d\mathsf D_{\mathrm{tar}}^{\theta,\pi;\eta}\right)^2
\le\left(\int 1^2\,d\mathsf D_{\mathrm{tar}}^{\theta,\pi;\eta}\right)
\left(\int (\mathfrak b^{\theta,\pi})^2\,d\mathsf D_{\mathrm{tar}}^{\theta,\pi;\eta}\right)
=H\|\mathfrak b^{\theta,\pi}\|_{2,\mathsf D_{\mathrm{tar}}^{\theta,\pi;\eta}}^2.
\]
Dividing by \(\|\mathfrak b^{\theta,\pi}\|_{2,\mathsf D_{\mathrm{log}}^\theta}^2\) and taking the supremum over \(\theta\in\Theta_s\) proves the claim, with the stated zero-denominator conventions.
\end{proof}

\begin{remark}[Scalar versus \(L^2\) Bellman coverage]
The coefficient \(\mathscr{C}\) controls the full target \(L^2\) Bellman-residual energy. The coefficient \(\kappa\) controls only the signed Bellman-residual integral that equals the value error. Therefore \(\mathscr{C}\) is the more literal Bellman-error coverage ratio, while \(\kappa\) is the sharper quantity for the interval proof.
\end{remark}

\section{Routing Experiment Details}
\label{app:experiments}

\subsection{Scope and relation to the theory}
\label{app:exp-scope}

The theoretical model allows finite horizons, transition learning, finite candidate classes, and planning over a confidence set.
The empirical study isolates the population-representation component.
It uses \(H=1\), deterministic reward labels, continuous neural reward models, and planning in one fitted reward game.
It therefore tests whether a task-relevant mean-field representation can be learned from finite offline population snapshots; it is not our aim to implement the confidence-set algorithm or validate the finite-sample bound of Section~\ref{sec:theory}.

Each row has an independently generated population context rather than the one fixed initial law appearing in the theoretical statement.
Conditional on that context, the focal state--action pair and the population snapshot are independent and have the observation structure of Section~\ref{sec:offline}.
The study thus measures generalization across population laws while preserving separate finite-\(N\) and finite-\(K\) effects.

\subsection{Environment and offline context distribution}
\label{app:exp-routing-environment}
\label{app:exp-routing-data}

The edge set is ordered deterministically: first the 20 origin-to-\(U\) edges, then the 16 \(U\)-to-\(V\) edges, and finally the 20 \(V\)-to-destination edges.
For \(x=(i,j)\) and \(a=(u,v)\), the vector \(f^\star(x,a)\) has ones at the coordinates of \(O_i\to U_u\), \(U_u\to V_v\), and \(V_v\to D_j\), and zeros elsewhere.
Consequently, the edge loads within each of the three layers sum to one for every population law.

The latency and route reward are given in \eqref{eq:exp-routing-latency}--\eqref{eq:exp-routing-reward}.
Since each edge load lies in \([0,1]\), the normalization
\begin{equation}
\label{eq:exp-routing-cmax}
C_{\max}=\max_{x,a}\left\{\delta_{x,a}+
\sum_{e=1}^{56} f_e^\star(x,a)
\left(\tau_e+\frac{\alpha_e}{c_e}+\frac{\beta_e}{c_e^2}\right)\right\}
\end{equation}
guarantees \(r^\star\in[0,1]\).
We draw the asymmetric edge parameters and route offsets once and hold the resulting network fixed across all training seeds and all \((N,K)\)-coordinates.

\paragraph{Demand law.}
A context first generates product-form origin--destination demand.
Let \(g^O,g^D\in\RR^5\) have independent standard-normal coordinates, and let \(T_O,T_D>0\) be independent temperatures.
With a fixed uniform floor \(\varepsilon_\mu\), set
\begin{align}
q_i^O &=(1-\varepsilon_\mu)
\frac{\exp(g_i^O/T_O)}{\sum_{i'}\exp(g_{i'}^O/T_O)}
+\frac{\varepsilon_\mu}{5},\\
q_j^D &=(1-\varepsilon_\mu)
\frac{\exp(g_j^D/T_D)}{\sum_{j'}\exp(g_{j'}^D/T_D)}
+\frac{\varepsilon_\mu}{5},
\end{align}
and \(\mu_c(i,j)=q_i^Oq_j^D\).
The product form is deliberate: the raw population mean identifies the complete demand law through its origin and destination marginals.

\paragraph{Behavior policy.}
The behavior population uses structured route preferences.
Draw independent standard-normal arrays \(\xi^{OU}_{i,u}\), \(\xi^{UV}_{u,v}\), and \(\xi^{VD}_{v,j}\), and a route-choice temperature \(T_\rho>0\).
Define
\begin{equation}
\label{eq:exp-routing-behavior-logit}
s_{i,j,u,v}=-\kappa_{\mathrm{ff}}\bigl(\tau_{i,u}+\tau_{u,v}+\tau_{v,j}\bigr)
+\sigma_{OU}\xi^{OU}_{i,u}
+\sigma_{UV}\xi^{UV}_{u,v}
+\sigma_{VD}\xi^{VD}_{v,j},
\end{equation}
and, with exploration floor \(\varepsilon_\rho\),
\begin{equation}
\label{eq:exp-routing-behavior-policy}
\rho_c(u,v\mid i,j)=(1-\varepsilon_\rho)
\frac{\exp(s_{i,j,u,v}/T_\rho)}
{\sum_{u',v'}\exp(s_{i,j,u',v'}/T_\rho)}
+\frac{\varepsilon_\rho}{16}.
\end{equation}
The temperatures are drawn from fixed log-uniform ranges.
These ranges, the exploration floors, and all scalar coefficients are fixed across models and resource coordinates.

\paragraph{Rows and nested resource grids.}
For each context, we compute \(\mu_c\), \(\rho_c\), \(\nu_c=\mu_c\otimes\rho_c\), and \(m_c^\star\) analytically.
We then draw one focal pair and an independent ordered stream of population samples.
We experiment with \(N\in\{1000,5000,25000,100000,200000\}\) and \(K\in\{1,2,4,8,16,32,64,128,256,512,1024\}\).
Smaller \(N\)-datasets are prefixes of a common row stream, and smaller \(K\)-snapshots are prefixes of the same ordered population stream within each row.
Training, validation, and final evaluation use disjoint random-number streams.

\subsection{Reward-model architectures and fitting}
\label{app:exp-routing-models}

The common focal encoder uses the 26-dimensional encoding \(z(x,a)\) and has shape \(26\to64\to64\to32\), with SiLU nonlinearities.
The learned model's encoder for each population sample and the monolithic population network have shape \(26\to64\to64\to56\).
The oracle adapter has shape \(56\to64\to64\to56\).
After concatenating the focal and population embeddings, the common reward head has shape \(88\to128\to64\to1\), followed by a sigmoid.
The single-agent network is widened to keep its total parameter count comparable.

The full-population-law branch receives the 400-entry empirical histogram and has shape \(400\to20\to20\to56\).
Under the count-only parameter-matching rule, the complete full-law model has 37297 parameters, compared with 37209 for the learned mean-field model, a difference of 0.237\%.
It is trained independently at every one of the 66 \((N,K)\)-coordinates.

Every model minimizes scalar reward mean-squared error with AdamW, learning rate \(10^{-3}\), weight decay \(10^{-5}\), effective row batch size 256, gradient-norm clipping at one, validation every 250 updates, and a maximum of 40000 optimizer updates.
The selected checkpoint is the strict minimum of validation MSE under the model's actual training input: no population input for the single-agent model, finite snapshots for the monolithic, learned, finite-\(K\)-oracle, and full-law models, and the analytic representation for the infinite-population oracle.
Final population RMSE and true Nash gap are not available to the trainer or checkpoint selector.

For every fitted checkpoint not reused from an existing grid coordinate, validation MSE attained its minimum at the final scheduled evaluation, so these fits share a common fixed update budget.

\subsection{Exact reward evaluation}
\label{app:exp-routing-evaluation}

For an evaluation context \(c\), all 400 focal categories are enumerated.
The exact population inputs are
\begin{align}
\overline m_c^{\mathrm{learned}} &=\sum_{x,a}\nu_c(x,a)f^\theta(x,a),\\
\overline m_c^{\mathrm{mono}} &=g^\theta\!\left(\sum_{x,a}\nu_c(x,a)z(x,a)\right),\\
\overline m_c^{\mathrm{oracle}\text{-}K} &=\overline m_c^{\mathrm{oracle}\text{-}\infty}=m_c^\star,\\
\overline m_c^{\mathrm{full}} &=h^\theta(\nu_c),
\end{align}
where \(g^\theta\) is the monolithic population network and \(h^\theta\) is the full-law population network.
The single-agent model has no population branch.
Thus the learned model is evaluated using the expected encoding of a population sample under \(\nu_c\), while the full-law model is evaluated at the exact 400-entry law; neither is handed the edge-load oracle unless that input is part of the model's definition.

The inner sum in~\eqref{eq:exp-routing-rmse} is exact and accumulated in float64.
The outer expectation over contexts is approximated by eight independently Owen-scrambled Sobol blocks shared across models and seeds.

\subsection{Fitted-game planning and exact controls}
\label{app:exp-routing-planning}
For a target demand law \(\mu\) and a fitted reward model, a candidate policy \(\eta\) is evaluated using that model's exact population input under \(\mu\otimes\eta\).
Let \(\widehat r_\theta^\eta(x,a)\) denote the resulting reward table.
The fitted-model residual is
\begin{equation}
\label{eq:exp-routing-model-gap}
\widehat\gap_{\mu,\theta}(\eta)=\sum_x\mu(x)\left[
\max_a\widehat r_\theta^\eta(x,a)-
\sum_a\eta(a\mid x)\widehat r_\theta^\eta(x,a)
\right].
\end{equation}
This quantity can be computed without consulting the true environment.

We apply the same 12-candidate optimization portfolio to every fitted reward model.
Entropic mirror descent and entropic Mirror--Prox each start from the uniform policy and use step sizes in \(\{5,20,30\}\), producing six candidates.
The other six candidates minimize a smoothed Nash-gap objective with Adam, using learning rates in \(\{0.03,0.1\}\), temperatures \(0.1,0.03,0.01\) in sequence, and three deterministic initial policies.
Each candidate optimization uses at most 5000 iterations, checks the fitted-model residual every 10 iterations, and stops early when that residual is at most \(10^{-5}\).
These settings and selection rules are applied uniformly across all fitted models.
The entropic mirror-descent update is
\begin{equation}
\label{eq:exp-routing-md-update}
\tilde\eta_{t+1}(a\mid x)=
\frac{\eta_t(a\mid x)\exp(\gamma_t\widehat r_\theta^{\eta_t}(x,a))}
{\sum_b\eta_t(b\mid x)\exp(\gamma_t\widehat r_\theta^{\eta_t}(x,b))}.
\end{equation}
For each checkpoint and target demand, the returned candidate is the one with the smallest fitted-model residual~\eqref{eq:exp-routing-model-gap}.
The true reward is never used for candidate selection.
The selected policy is then evaluated under the true reward via~\eqref{eq:exp-routing-gap}.

The \emph{uniform policy} is \(\eta_{\mathrm{unif}}(a\mid x)=1/16\).
The \emph{exact-reward control} applies the same 12-candidate portfolio after replacing the learned reward by the analytic reward in~\eqref{eq:exp-routing-reward}.
It isolates the optimization error of the common planner from reward-model error.

\subsection{Statistical protocol and reproducibility}
\label{app:exp-routing-statistics}

We compute uncertainty across the five training seeds.
Within a seed, models share contexts, focal samples, ordered population streams, nested resource prefixes, and the fixed quasi--Monte Carlo context samples used for evaluation.
For each seed, we first average the metric over all of these evaluation context samples.
If \(L_s\) is the resulting metric for seed \(s\in\{0,\ldots,4\}\), we report
\[
\overline L\pm t_{0.975,4}\frac{\operatorname{sd}(L_0,\ldots,L_4)}{\sqrt5}.
\]
Thus every ``\(\pm\)'' quantity in the primary landscape is a 95\% Student-\(t\) half-width over five paired training seeds, not an integration error bar.
QMC integration uncertainty is assessed separately from variation across training seeds.

\end{document}